\documentclass[11pt,a4paper]{article}

\usepackage{bm, amsmath, amsthm, amssymb, natbib}
\usepackage[pdftex]{graphicx}
\usepackage{color}

\newtheorem{Definition}{Definition}[section]
\newtheorem{Theorem}{Theorem}[section]
\newtheorem{Lemma}{Lemma}[section]
\newtheorem{Proposition}{Proposition}[section]
\newtheorem{Corollary}{Corollary}[section]

\newcommand{\conv}{\operatorname{conv}}
\renewcommand{\indent}{\hspace*{5mm}}

\begin{document}

\title{Game-theoretic derivation of upper hedging prices of multivariate contingent claims and submodularity}
\author{Takeru Matsuda\thanks{Graduate School of Information Science and Technology, The University of Tokyo}
\ and Akimichi Takemura\thanks{Faculty of Data Science, Shiga University}}
\date{}

\maketitle

\begin{abstract}
We investigate upper and lower hedging prices of multivariate contingent claims from the viewpoint of game-theoretic probability and submodularity.
By considering a game between ``Market'' and ``Investor'' in discrete time, the pricing problem is reduced to a backward induction of an optimization over simplexes.
For European options with payoff functions satisfying a combinatorial property called submodularity or supermodularity,
this optimization is solved in closed form by using the Lov\'asz extension and the upper and lower hedging prices can be calculated efficiently.
This class includes the options on the maximum or the minimum of several assets.
We also study the asymptotic behavior as the number of game rounds goes to infinity.
The upper and lower hedging prices of European options converge to the solutions of the Black-Scholes-Barenblatt equations.
For European options with submodular or supermodular payoff functions, 
the Black-Scholes-Barenblatt equation is reduced to the linear Black-Scholes equation and it is solved in closed form.
Numerical results show the validity of the theoretical results.
\end{abstract}

\section{Introduction}
\label{sec:intro}
The pricing of contingent claims is a central problem in mathematical finance \citep{Karatzas}.
The fundamental models of financial markets are the binomial model \citep{Shreve} and the geometric Brownian motion model \citep{Shreve2} 
in discrete-time and continuous-time setting, respectively.
These models describe complete markets and therefore the price of any contingent claim is obtained by arbitrage argument.
Specifically, the Cox-Ross-Rubinstein formula \citep{Cox} and the Black-Scholes formula \citep{Black} provide the exact price
in the binomial model and the geometric Brownian motion model, respectively. 
These formulas are derived by constructing a hedging portfolio for the seller to replicate the contingent claim.

In general, the market is incomplete and the above formulas are not applicable.
Even in incomplete markets, we can define the upper and lower hedging prices of a contingent claim by considering superreplication \citep{Karatzas}.
\cite{Musiela} and \cite{ElKaroui} provide fundamental results on the upper hedging price 
in discrete-time models and continuous-time models, respectively. 
As a special case, for discrete-time models with bounded martingale differences, 
\cite{Ruschendorf} pointed out the upper hedging price of a convex contingent claim is obtained by the extremal binomial model.

Although the above studies focused on contingent claims depending on a single asset, 
there are contingent claims for which the payoff depends on two or more assets \citep{Stapleton}, which are called multivariate contingent claims.
For example, \cite{Stulz} and \cite{Johnson} investigated the pricing of options on the maximum or the minimum of several assets.
\cite{Boyle} developed a numerical method for pricing multivariate contingent claims in discrete-time models.
Although their method is based on a lattice binomial model that is originally incomplete,
they change the model to make it complete by specifying correlation coefficients between all the pairs of assets.
Thus, their method does not compute the upper hedging price. 
On the other hand, \cite{Romagnoli} considered superreplication in continuous-time models and derived the pricing formula based on the Black-Scholes-Barenblatt equation.
They also provided some sufficient conditions on payoff functions for reduction of the Black-Scholes-Barenblatt equation to the linear Black-Scholes equation.

Whereas existing studies on the upper and lower hedging price are based on stochastic models of financial markets,
\cite{Nakajima} investigated this problem from the viewpoint of the game-theoretic probability \citep{Shafer},
in which only the protocol of a game between ``Investor'' and ``Market'' is formulated without specification of a probability measure. 
They showed that the upper hedging price in the discrete-time multinomial model is obtained by a backward induction of linear programs,
and that the upper hedging price of a European option converges to the solution of the one-dimensional Black-Scholes-Barenblatt equation as the number of game rounds goes to infinity.

In this paper, we investigate the upper hedging price of multivariate contingent claims by extending the game-theoretic probability approach of \cite{Nakajima}.
We consider a discrete-time multinomial model with several assets 
and show that the upper hedging price of multivariate contingent claims is given by a backward induction of a maximization problem 
whose domain is a set of simplexes, which becomes intractable in general as the number of assets increases.
However, we find that this maximization is solved in closed form if the contingent claim satisfies a combinatorial property called submodularity or supermodularity \citep{Fujishige}.
Specifically, the maximizing simplex is determined by using the Lov\'asz extension \citep{Lovasz} for European options with supermodular payoff function on two or more assets and also European options with submodular payoff function on two assets.
As realistic examples, we prove that options on the maximum and the minimum of several assets are submodular and supermodular, respectively.
Then, by considering the asymptotics of the number of game rounds, we show that the upper hedging price of a European option converges to the solution of the Black-Scholes-Barenblatt equation.
In particular, for European options with supermodular payoff function on two or more assets and also European options with submodular payoff function on two assets, 
the Black-Scholes-Barenblatt equation reduces to the linear Black-Scholes equation, which is solved in closed form.
Finally, we confirm the validity of the theoretical results by numerical experiments.

As in \cite{Nakajima}, we consider the price processes in an additive form and 
the Black-Scholes-Barenblatt equation in section \ref{sec:limiting} is actually an additive form
of the Black-Scholes-Barenblatt equation in the standard finance literature.
Similarly, the linear Black-Scholes equation is given in the form of a heat equation.
However, the results of this paper holds for the usual multiplicative model with simple exponential
transformations.  Except for a few places we do not specifically indicate that our equations 
are in the additive form.

In section 2, we provide a formulation of the upper hedging price based on the game-theoretic probability.
In section 3, we derive results for the special case of European options with submodular or supermodular payoff function, which include the option on the maximum or the minimum.
In section 4, we derive PDE for asymptotic upper hedging prices.
In section 5, we confirm the theoretical results by numerical experiments.
In section 6, we give some concluding remarks and discuss future works.

\section{Game-theoretic derivation of upper hedging prices of multivariate contingent claims}
In this section, we present a formulation of the upper hedging price based on the game-theoretic probability \citep{Shafer}.
The pricing problem is reduced to a backward induction of linear programs.
As special cases, we consider convex or separable payoff functions.

\subsection{Definitions and notation}
Let $\chi=\{ a_1,\cdots,a_l \} \subset \mathbb{R}^d$, $l\ge d+1$,  be a finite set, which we call a \textit{move set}.
Let 
$\conv\chi$ denote the convex hull of $\chi$. We assume
that the dimension of  $\conv\chi$ is $d$ and 
$\conv\chi$ 
contains the origin in its interior.
The protocol of the $N$-round multinomial game with $d$ assets is defined as follows:
\begin{quote}
\indent ${\mathcal K}_0 := \alpha$ \\
\indent FOR $n=1, 2, \cdots, N$\\
\indent\indent Investor announces $M_n \in \mathbb{R}^d$\\
\indent\indent Market announces $x_n \in \chi$\\
\indent\indent ${\mathcal K}_n := {\mathcal K}_{n-1} + M_n^{\top} x_n$\\
\indent END FOR
\end{quote}
Here, $\alpha$ denotes the initial capital of Investor, $M_n$ corresponds to the vector of
amounts Investor buys the assets, $x_n$ corresponds to the vector of price
changes of the assets and ${\mathcal K}_n$ corresponds to Investor's capital at the end of round $n$.
When $d=1$, this game reduces to the game analyzed in \cite{Nakajima}.
One natural candidate for $\chi$ is a product set 
\begin{equation}
	\chi=\{ a^{(1)}_1, \cdots, a^{(1)}_{n_1} \} \times \cdots \times \{ a^{(d)}_1, \cdots, a^{(d)}_{n_d} \}, \label{prod_chi}
\end{equation}
where $a^{(1)}_1 < \cdots < a^{(1)}_{n_1}, \cdots, a^{(d)}_1 < \cdots < a^{(d)}_{n_d}$.
Such $\chi$ with $n_1=\cdots=n_d=2$ was adopted by the lattice binomial model \citep{Boyle}.

We call $\chi^N$ the \textit{sample space} and $\xi = x_1 \cdots x_N \in \chi^N$ a \textit{path} of Market's moves.
For $n=1,\cdots,N$, $\xi^n = x_1 \cdots x_n \in \chi^n$ is a partial path.
The initial empty path $\xi^0$ is denoted as $\square$.
We call $\mathcal{S}: x_1 \cdots x_{n-1} \mapsto M_n$ a \textit{strategy}.
When Investor adopts the strategy $\mathcal{S}$, his capital at the end of round $n$ is given by $\alpha+\mathcal{K}_n^{\mathcal{S}} (\xi^n)$, where
\begin{equation*}
	\mathcal{K}_n^{\mathcal{S}} (\xi^n) = \sum_{i=1}^n \mathcal{S} (\xi^{i-1})^{\top} x_i.
\end{equation*}

We call a function $f: \chi^N \to \mathbb{R}$ a \textit{payoff function} or a \textit{contingent claim}.
If $f$ depends only on $S_N=x_1+\cdots+x_N$, then $f$ is called a \textit{European option}.
The \textit{upper hedging price} (or simply upper price) of $f$ is defined as
\begin{equation}
	\bar{E}_{\chi} (f) = \inf \{ \alpha \mid \exists \mathcal{S}, \alpha + \mathcal{K}_N^{\mathcal{S}} (\xi) \geq f(\xi), \forall \xi \in \chi^N \} \label{upper_def}
\end{equation}
and the \textit{lower hedging price} is defined as
\begin{equation*}
	\underline{E}_{\chi} (f) = -\bar{E}_{\chi} (-f).
\end{equation*}
A strategy attaining the infimum in \eqref{upper_def} is called a \textit{superreplicating strategy} for $f$.

A market is called \textit{complete} if the upper hedging price and the lower hedging price coincide for any payoff function.
For example, the binomial model, which corresponds to $d=1$ and $|\chi|=2$, is complete \citep{Shreve}.

As discussed in section \ref{sec:intro} we consider an additive form of the game where the prices changes $x_n$'s are added rather than multiplied in $S_N=x_1+\cdots+x_N$.

\subsection{Formulation with linear programming}
Following \cite{Nakajima}, we formulate the pricing problem as a recursive linear programming.

First, we consider the single-round game ($N=1$).
Note that the payoff function is $f: \chi \to \mathbb{R}$.
Let $\Gamma=\{ \widetilde{\chi} \subset \chi \mid |\widetilde{\chi}| = d+1, \ 0 \in \conv \widetilde{\chi}, \ \dim \conv \widetilde{\chi}=d \}$ be the set of simplexes of dimension $d$ containing the origin.
For each $\widetilde{\chi} = \{ a_{i_0},\cdots,a_{i_d} \} \in \Gamma$, define
\begin{equation*}
	I(\widetilde{\chi},f) = \sum_{j=0}^d p^{\widetilde{\chi}}_j f(a_{i_j}),
\end{equation*}
where 
the probability vector $p^{\widetilde{\chi}} = (p^{\widetilde{\chi}}_0,\cdots,p^{\widetilde{\chi}}_d)$ is defined as the unique solution of the linear equations
\begin{equation}
	\sum_{j=0}^d p^{\widetilde{\chi}}_j = 1, \quad \sum_{j=0}^d p^{\widetilde{\chi}}_j a_{i_j,k} = 0 \ (k=1,\cdots,d). \label{chi_prob}
\end{equation}
By extending Proposition 2.1 of \cite{Nakajima}, we obtain the following.

\begin{Proposition}\label{prop_dual}
For a single-round game, the upper hedging price of a payoff function $f$ is given by
\begin{equation}
	\bar{E}_{\chi} (f) = \max_{\widetilde{\chi} \in \Gamma} I(\widetilde{\chi},f). \label{single_upper}
\end{equation}
\end{Proposition}
\begin{proof}
From \eqref{upper_def}, the upper hedging price $\bar{E}_{\chi} (f)$ of $f$ is the optimal value of the following linear program:
\begin{align*}
	\min_{\alpha,M} \quad & \begin{pmatrix} 1 & 0 & \cdots & 0 \end{pmatrix} \begin{pmatrix} \alpha \\ M_1 \\ \vdots \\ M_d \end{pmatrix} \\
	{\rm s.t.} \quad & \begin{pmatrix} 1 & a_{1,1} & \cdots & a_{1,d} \\ 1 & a_{2,1} & \cdots & a_{2,d} \\ \vdots \\ 1 & a_{l,1} & \cdots & a_{l,d} \end{pmatrix} \begin{pmatrix} \alpha \\ M_1 \\ \vdots \\ M_d \end{pmatrix} \geq \begin{pmatrix} f(a_1) \\ \vdots \\ f(a_l) \end{pmatrix} .
\end{align*}
The dual of this linear program is
\begin{align*}
	\max_p \quad & \begin{pmatrix} f(a_1) & \cdots & f(a_l) \end{pmatrix} p \\
	{\rm s.t.} \quad & \begin{pmatrix} 1 & \cdots & 1 \\ a_{1,1} & \cdots & a_{l,1} \\ \vdots \\ a_{1,d} & \cdots & a_{l,d} \end{pmatrix} p = \begin{pmatrix} 1 \\ 0 \\ \vdots \\ 0 \end{pmatrix}, \quad p \geq 0.
\end{align*}
From the complementary condition \citep{Boyd}, 
there exists an optimal solution $p^{*}$ of the dual problem that has at most $d+1$ nonzero variables.
Then, from the constraint of the dual problem, we have $p^{*}=p^{\widetilde{\chi}}$ for some $\widetilde{\chi}$.
Therefore, we obtain \eqref{single_upper}.
\end{proof}

Similarly, the lower hedging price of $f$ for a single-round game is
\begin{equation*}
	\underline{E}_{\chi} (f) = \min_{\widetilde{\chi} \in \Gamma} I(\widetilde{\chi},f).
\end{equation*}

The relation \eqref{single_upper} is interpreted as follows.
A probability vector $p=(p_1,\cdots,p_l)$ is called a \textit{risk neutral measure} on $\chi$ if the expectation under $p$ is zero:
\begin{equation*}
	\sum_{j=1}^l p_j a_{j,k} = 0 \ (k=1,\cdots,d).
\end{equation*}
Let $\mathcal{P}(\chi)$ be the set of risk neutral measures on $\chi$.
Note that $\mathcal{P}(\chi)$ is closed and convex.
Then, the set $\{ p^{\widetilde{\chi}} \mid \widetilde{\chi} \in \Gamma \}$ coincides with the set of extreme points (vertices) of $\mathcal{P}(\chi)$.
Since the maximum of a linear function on a closed convex set is attained at extreme points,
the maximization in \eqref{single_upper} is interpreted as searching over all risk neutral measures:
\[
	\bar E_\chi(f)=\max_{p\in {\cal P}(\chi)} E_p(f),
\]
where $E_p$ denotes the expectation with respect to $p$.

When $d=2$ and $\chi=\{ a^{(1)}_1, a^{(1)}_{2} \} \times \{ a^{(2)}_1, a^{(2)}_{2} \}$, 
the maximization in \eqref{single_upper} involves two candidates of $\widetilde{\chi}$.
Specifically, the possible $\widetilde{\chi}$ is ABC or ABD in Figure \ref{LP_2d}.
When $d \geq 3$ and $\chi=\{ a^{(1)}_1, a^{(1)}_{2} \} \times \cdots \times \{ a^{(d)}_1, a^{(d)}_{2} \}$, 
the maximization in \eqref{single_upper} becomes more difficult since the number of possible $\widetilde{\chi}$ becomes large.
For example, when $d=3$, the number of possible $\widetilde{\chi}$ can be as large as 14, as shown in Appendix A.
In general, the number of possible $\widetilde{\chi}$ is at least $2^{d-2}$, as shown in Appendix B.
In the next section, we will provide some sufficient conditions on $f$ for the maximization in \eqref{single_upper} to be solved explicitly.

\begin{figure}
	\centering
	\includegraphics[scale=0.3]{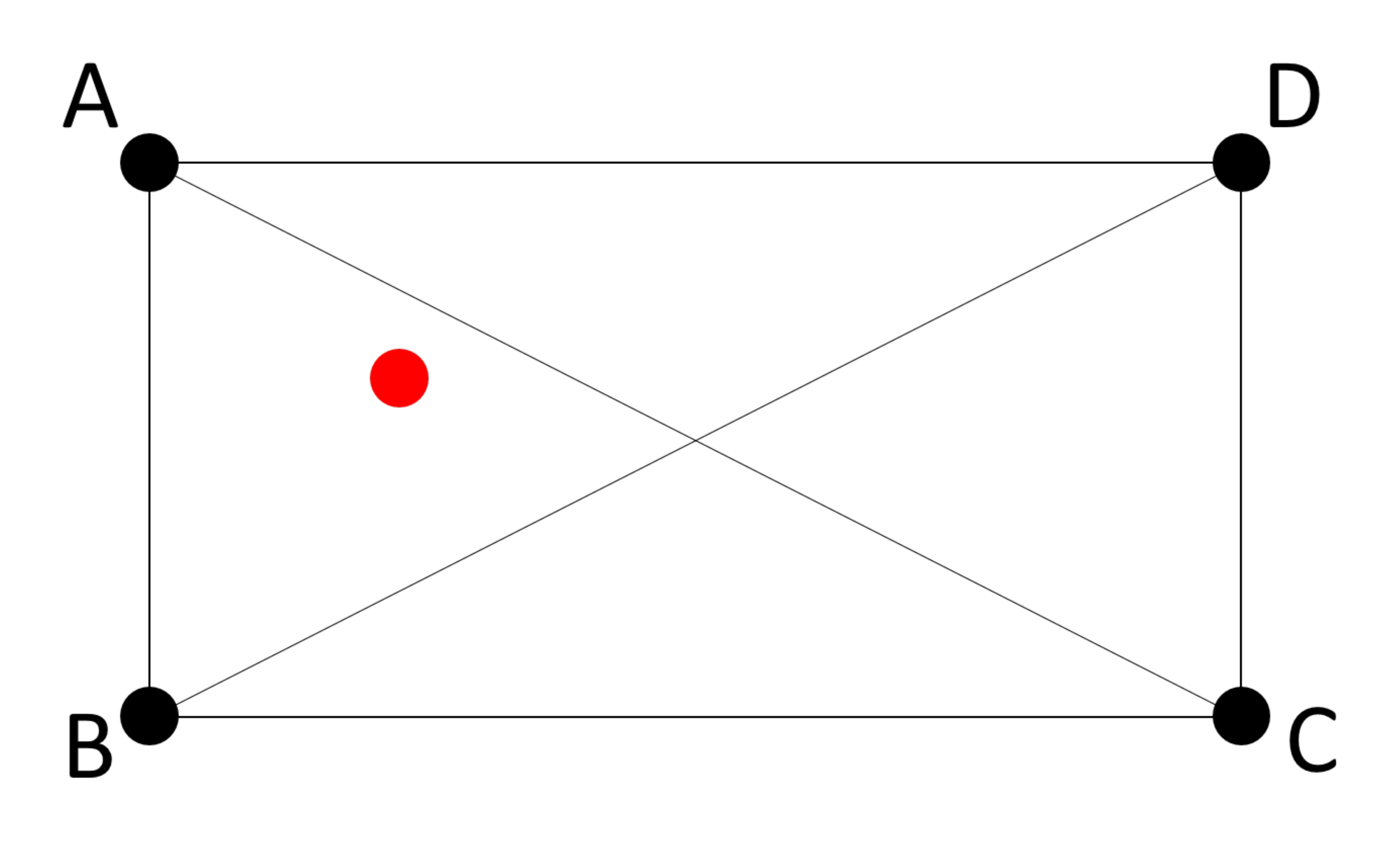}
	\caption{Single-round game when $d=2$. The move set $\chi$ corresponds to the four points A-D. The origin is  the point inside the rectangle drawn in red.}
	\label{LP_2d}
\end{figure}


Now, we consider the $N$-round game. 
As discussed by \cite{Nakajima} for $d=1$, the upper hedging price is calculated by solving linear programs recursively.
Specifically, let $\bar{f} (\cdot, N-n): \chi^n \to \mathbb{R}$, $n=N,N-1,\cdots,0$ be given by
\begin{equation}
	\bar{f} (\xi^n, N-n) = \max_{\widetilde{\chi} \in \Gamma} I(\widetilde{\chi}, \bar{f}(\xi^n \cdot, N-n-1)), \label{N_upper}
\end{equation}
with the initial condition $\bar{f} (\xi, 0) = f(\xi)$ for $\xi \in \chi^N$.
Here, $\bar{f}(\xi^n \cdot, N-n-1)$ denotes a function $\chi \to \mathbb{R}$ defined by $a_i \mapsto \bar{f}(\xi^n a_i, N-n-1)$.
Then, we have the following result.

\begin{Proposition}
The upper hedging price of $f: \chi^N \to \mathbb{R}$ is given by
\begin{equation*}
	\bar{E}_{\chi} (f) = \bar{f} (\square, N).
\end{equation*}
\end{Proposition}

Therefore, the upper hedging price $\bar{E}_{\chi} (f)$ is calculated by a backward induction of \eqref{N_upper}.
In particular, for a European option, $\bar{f} (\xi^n,N-n)$ depends only on $S_n=x_1+\cdots+x_n$ 
and so the required number of calculations \eqref{N_upper} grows only polynomially with $N$.

The lower hedging price $\underline{E}_{\chi} (f)$ is also calculated by a backward induction.
In general, a market with $|\chi|=d+1$ is complete since $|\Gamma| = 1$, irrespective of $f$.


\subsection{Convex payoff functions}
Here, we consider the case where the payoff function $f$ is a European option $f(\xi)=F(S_N)$ with convex function $F$.
When $d=1$, the calculation of the upper hedging price of a convex payoff function is reduced to the binomial model \citep{Ruschendorf}, 
since the maximization in \eqref{single_upper} is always uniquely attained by the extreme pair $\widetilde{\chi} = \{ \min \chi, \max \chi \}$.
For general $d$, we have the following result.

\begin{Proposition}\label{prop_convex}
Suppose that $f$ is a European option $f(\xi)=F(S_N)$ with convex function $F$.
Let $\chi_0$ be the set of vertices of $\conv\chi$. 
Then, the upper hedging price of $f$ with move set $\chi$ coincides with that with move set $\chi_0$:
\begin{equation}
	\bar{E}_{\chi} (f) = \bar{E}_{\chi_0} (f). \label{convex_eq}
\end{equation}
\end{Proposition}
\begin{proof}
We provide a proof for the single-round game.
By applying it to each induction step, the proof for the multi-round game is obtained.

Since $\mathcal{P} (\chi_0) \subset \mathcal{P} (\chi)$, we obtain
\begin{equation}
	\bar E_\chi(f) = \max_{p\in {\cal P}(\chi)} E_p(f) \ge \max_{p\in {\cal P}(\chi_0)} E_p(f) = \bar E_{\chi_0}(f). \label{eq:gege}
\end{equation}

Let $\chi_0=\{ b_1,\cdots,b_m \}$.
From Caratheodory's Theorem \citep{Rockafeller}, for each $a_i \in \chi \subset {\rm conv} \chi_0$, 
there exists a subset $\{ b_{j_0(i)},\cdots,b_{j_d(i)} \}$ of $\chi_0$
and a probability vector $(q_0^{i}, \cdots, q_d^i)$ such that
\[
	a_i = \sum_{k=0}^d q_k^{i} b_{j_k(i)}.
\]
Thus, from the convexity of $f$,
\begin{equation*}
	f(a_i) \leq \sum_{k=0}^d q_k^{i} f(b_{j_k(i)}).
\end{equation*}
Therefore,
\[
	E_p(f) = \sum_{i=1}^l p_i f(a_i) \le \sum_{i=1}^l \sum_{k=0}^d p_i q_k^{i} f(b_{j_k(i)})
\]
for any $p \in \mathcal{P} (\chi)$.
By summing up the right hand side for each point $b_j$ of $\chi_0$, we obtain
\[
	E_p(f) \leq \sum_{j=1}^m r_j f(b_j),
\]
where $r \in \mathcal{P} (\chi_0)$ since
\[
	\sum_{j=1}^m r_j = \sum_{i=1}^l \sum_{k=0}^d p_i q_k^{i} = \sum_{i=1}^l p_i  = 1
\]
and
\[
	\sum_{j=1}^m r_j b_j = \sum_{i=1}^l \sum_{k=0}^d p_i q_k^{i} b_{j_k(i)} = \sum_{i=1}^l p_i a_i = 0.
\]
Thus, for every $p \in \mathcal{P} (\chi)$, there exists $r \in \mathcal{P} (\chi_0)$ such that
\[
	E_p(f) \le E_r(f)  \le \bar E_{\chi_0} (f).
\]  
By taking the maximum in the left hand side, we obtain
\begin{equation}
  \bar E_{\chi}(f) \le \bar E_{\chi_0} (f). \label{eq:lele}
\end{equation}

From \eqref{eq:gege} and \eqref{eq:lele}, we obtain \eqref{convex_eq}.
\end{proof}

Suppose $d \geq 2$ and the move set $\chi$ is a product set \eqref{prod_chi}.
Then, Proposition \ref{prop_convex} implies that we can redefine 
$\chi_0 = \{ a^{(1)}_1, a^{(1)}_{n_1} \} \times \cdots \times \{ a^{(d)}_1, a^{(d)}_{n_d} \}$ as the move set.
However, unlike the case $d=1$, even with this reduction to $\chi_0$, the maximizing $\widetilde{\chi}$ in \eqref{single_upper} is not unique in general.
One example is the option on the maximum with $d \geq 3$, as we will see in the next section.

\subsection{Separable payoff functions}
We call a European option $f$ \textit{separable} if it can be decomposed as
\begin{equation}
	f(\xi) = F(S_N) = \sum_{k=1}^K f_k ((S_N)_{(k)}), \label{separable}
\end{equation}
where $\{ (S_N)_{(1)},\cdots,(S_N)_{(K)} \}$ is a partition of $\{ (S_N)_1,\cdots,(S_N)_d \}$:
\begin{equation*}
	(S_N)_{(k)} = ((S_N)_i)_{i \in A_k}, \quad \bigcup_{k=1}^K A_k = \{ 1,\cdots,d \}, \quad A_k \cap A_l = \emptyset \ (k\neq l).
\end{equation*}
We assume that the move set $\chi$ is a direct product of move sets $\chi_{(k)}$ for each subset $A_k$.
For example, when $K=d$ and $A_k = \{ k\}$, the move set $\chi$ is a product set $\chi=\chi_1 \times \cdots \times \chi_d$ as in \eqref{chi_prob}.
Then, the calculation of the upper hedging price of a separable European option is reduced to that for each component European option as follows.

\begin{Proposition}\label{prop_sep}
Suppose $f$ is separable \eqref{separable} and $\chi$ is a direct product of move sets $\chi_{(k)}$ for each subset $A_k$.
Then, the upper hedging price of $f$ coincides with the sum of the upper hedging prices of $f_1,\cdots,f_K$:
\begin{equation}
	\bar{E}_{\chi} (f) = \sum_{k=1}^K \bar{E}_{\chi_{(k)}} (f_k). \label{eq_prop_sep}
\end{equation}
\end{Proposition}
\begin{proof}
We provide a proof for the single-round game.
By applying it to each round, the proof for the multi-step game is obtained.

For simplicity of notation we consider the case $K=d$ and $A_k = \{ k\}$ without essential loss of generality.
Then, the move set $\chi$ is a product set $\chi=\chi_1 \times \cdots \times \chi_d$ and the payoff function is
written as
\begin{equation}
	f(\xi) = \sum_{k=1}^d f_k ((S_N)_k). \label{sep_f}
\end{equation}

Let $p$ be an arbitrary risk neutral measure on $\chi$ and let $p_1,\cdots,p_d$ be its marginals.
Then, each $p_k$ is a risk neutral measure on $\chi_k$.
By taking expectations in \eqref{sep_f},
\[
	E_p(f)= \sum_{k=1}^d E_{p_k}(f_k).
\]
Since $E_{p_k} (f_k) \leq \bar{E}_{\chi_k} (f_k)$,
\[
	E_p(f) \le \sum_{k=1}^d \bar{E}_{\chi_k} (f_k).
\]
Therefore,
\begin{equation}
	\bar{E}_{\chi} (f) = \max_{p \in \mathcal{P} (\chi)} E_p(f) \le \sum_{k=1}^d \bar{E}_{\chi_k} (f_k).   \label{eq:le}
\end{equation}

Conversely, let $p_1,\cdots,p_d$ be arbitrary risk neutral measures on $\chi_1,\cdots,\chi_d$, respectively.
Define $p=p_1 \times \cdots \times p_d$.
Then, $p$ is a risk neutral measure on $\chi$ and it satisfies
\[
	\sum_{k=1}^d E_{p_k}(f_k) = E_{p}(f).
\]
Since $E_{p} (f) \leq \bar{E}_{\chi} (f)$,
\[
	\sum_{k=1}^d E_{p_k}(f_k) \leq \bar{E}_{\chi} (f).
\]
Therefore,
\begin{equation}
	\sum_{k=1}^d \bar{E}_{\chi_k} (f_k) = \max_{p_1,\cdots,p_d} \sum_{k=1}^d E_{p_k}(f_k) \leq \bar{E}_{\chi} (f). \label{eq:ge}
\end{equation}

From \eqref{eq:le} and \eqref{eq:ge}, we obtain \eqref{eq_prop_sep}.
\end{proof}

\section{Submodular and supermodular payoff functions}
We have seen that the calculation of the upper hedging price reduces to backward induction of the maximization \eqref{single_upper}.
In this section, we show that this maximization is solved in closed form 
if the payoff function $f$ satisfies a combinatorial property called submodularity or supermodularity.
As a special case, we discuss the options on the maximum or the minimum.
Throughout this section, we assume that the move set $\chi$ is a product set \eqref{prod_chi} with $n_1=\cdots=n_d=2$, i.e., the lattice binomial model. 
Note that $a^{(k)}_1<0<a^{(k)}_2$ for $k=1,\cdots,d$, because we are assuming that
$\conv \chi$ contains the origin in its interior.

\subsection{Submodular and supermodular functions}
The concept of submodularity is fundamental in combinatorial optimization theory \citep{Fujishige}.

\begin{Definition}\label{def_submo}
\begin{itemize}
\item A set function $f: 2^X \to \mathbb{R}$ is said to be \textit{submodular} if it satisfies
\begin{equation*}
	f(U \cap V) + f(U \cup V) \leq f(U) + f(V),
\end{equation*}
for every $U,V \subset X$.

\item A set function $f: 2^X \to \mathbb{R}$ is said to be \textit{supermodular} if it satisfies
\begin{equation*}
	f(U \cap V) + f(U \cup V) \geq f(U) + f(V),
\end{equation*}
for every $U,V \subset X$.
\end{itemize}
\end{Definition}

It is well known (Theorem 44.1 of \cite{Schrijver}) that in Definition \ref{def_submo} we only need to consider $U$ and $V$ such that
\begin{equation}
|U\setminus V|=|V\setminus U|=1. \label{two-element-difference}
\end{equation}

We can extend the definition of submodularity and supermodularity to functions on the hypercube $[0,1]^d$ and $\mathbb{R}^d$.
For vectors $u,v \in \mathbb{R}^d$, we denote the vectors of componentwise minimum and maximum by $u \land v \in \mathbb{R}^d$ and $u \lor v \in \mathbb{R}^d$, respectively:
\begin{equation*}
	(u \land v)_i = \min (u_i,v_i), \quad (u \lor v)_i = \max (u_i,v_i).
\end{equation*}
Then, submodular and supermodular functions on the hypercube $[0,1]^d$ and $\mathbb{R}^d$ are defined as follows.

\begin{Definition}\label{def_conti}
\begin{itemize}
\item A function $\hat{f}: [0,1]^d \to \mathbb{R}$ or $\mathbb{R}^d \to \mathbb{R}$ is said to be \textit{submodular} if it satisfies
\begin{equation}
	\hat{f}(u \land v) + \hat{f}(u \lor v) \leq \hat{f}(u) + \hat{f}(v), \label{ineq1}
\end{equation}
for every $u,v \in [0,1]^d$ or $\mathbb{R}^d$.

\item A function $\hat{f}: [0,1]^d \to \mathbb{R}$ or $\mathbb{R}^d \to \mathbb{R}$ is said to be \textit{supermodular} if it satisfies
\begin{equation*}
	\hat{f}(u \land v) + \hat{f}(u \lor v) \geq \hat{f}(u) + \hat{f}(v),
\end{equation*}
for every $u,v \in [0,1]^d$ or $\mathbb{R}^d$.
\end{itemize}
\end{Definition}

We note that the concept of multivariate total positivity (MTP2) is closely related to submodularity \citep{Karlin,Fallat}.
Namely, a positive function is MTP2 if and only if its logarithm is supermodular.

When $\hat{f}: [0,1]^d \to \mathbb{R}$ or $\mathbb{R}^d \to \mathbb{R}$ is twice continuously differentiable, 
the submodularity and supermodularity of $\hat{f}$ are characterized by the signs of the mixed second order derivatives as follows.

\begin{Lemma}\label{lem_mod}
\begin{itemize}
\item A twice continuously differentiable function $\hat{f}: [0,1]^d \to \mathbb{R}$ or $\mathbb{R}^d \to \mathbb{R}$ is submodular if and only if
\begin{equation}
	\frac{\partial^2 \hat{f}}{\partial s_i \partial s_j} \leq 0 \label{partial_ineq}
\end{equation}
for every $i \neq j$.

\item A twice continuously differentiable function $\hat{f}: [0,1]^d \to \mathbb{R}$ or $\mathbb{R}^d \to \mathbb{R}$ is supermodular if and only if
\begin{equation*}
	\frac{\partial^2 \hat{f}}{\partial s_i \partial s_j} \geq 0
\end{equation*}
for every $i \neq j$.
\end{itemize}
\end{Lemma}
\begin{proof}
We only prove the first statement.
The proof of the second statement is similar by considering $-\hat{f}$.

Assume $\hat{f}$ is submodular and let $\varepsilon>0$ and $\varepsilon'>0$.
Let $e_i$ be the unit vector with $i$-th coordinate one and other coordinates zero.
Substituting $u=s + \varepsilon e_i$ and $v=s + \varepsilon' e_j$ into \eqref{ineq1},
\begin{equation*}
	\hat{f}(s) + \hat{f}(s + \varepsilon e_i + \varepsilon' e_j) \leq \hat{f}(s + \varepsilon e_i) + \hat{f}(s + \varepsilon' e_j).
\end{equation*}
Therefore,
\begin{equation*}
	\frac{\hat{f}(s + \varepsilon e_i + \varepsilon' e_j)-\hat{f}(s + \varepsilon' e_j)}{\varepsilon} \leq \frac{\hat{f}(s + \varepsilon e_i) - \hat{f}(s)}{\varepsilon}.
\end{equation*}
Putting $\varepsilon \to 0$, we obtain
\begin{equation*}
	\frac{\partial \hat{f}}{\partial s_i} (s + \varepsilon' e_j) \leq \frac{\partial \hat{f}}{\partial s_i} (s).
\end{equation*}
Therefore, we obtain \eqref{partial_ineq}.

Conversely, assume \eqref{partial_ineq}. 
To prove the submodularity of $\hat{f}$, it is sufficient to consider the case where $u$ and $v$ differ only in two elements as in \eqref{two-element-difference}. 
Without loss of generality, let $u=(a_1, b_2, c_3, \cdots, c_d)$ and $v=(b_1, a_2, c_3,\cdots, c_d)$, with $a_1 < b_1, a_2 < b_2$.
Then,
\begin{equation*}
	(\hat{f}(u \lor v) - \hat{f}(v)) - (\hat{f}(u) - \hat{f}(u \land v)) 
= \int_{a_1}^{b_1} \int_{a_2}^{b_2} \frac{\partial^2 \hat{f}}{\partial s_1 \partial s_2}(s_1, s_2,c_3,\cdots, c_d)
{\rm d} s_1 {\rm d} s_2 \leq 0.
\end{equation*}
Therefore, $\hat{f}$ is submodular.
\end{proof}

\subsection{Convex closure and Lov\'asz extension}
In considering the maximization \eqref{single_upper}, the concepts of convex closure and Lov\'asz extension are useful.
\cite{Dughmi} presents a brief survey on these topics.

Let $X = \{ x_1,\cdots,x_d \}$ be a finite set. 
For a subset $A$ of $X$, its characteristic vector $1_A \in \mathbb{R}^d$ is defined as
\begin{equation*}
	(1_A)_k = \begin{cases} 1 & (x_k \in A) \\ 0 & (x_k \not\in A) \end{cases}.
\end{equation*}
In the following, we identify $2^X$ with $\{ 0,1 \}^d$ by the bijection $A \mapsto 1_A$.
For a set function $f: 2^X \to \mathbb{R}$, a function $\hat{f}: [0,1]^d \to \mathbb{R}$ or $\mathbb{R}^d \to \mathbb{R}$ is said to be its \textit{extension} if it satisfies $\hat{f} (1_A) = f(A)$.

\begin{Definition}
For a set function $f: 2^X \to \mathbb{R}$, its convex closure $f^-: [0,1]^d \to \mathbb{R}$ and concave closure $f^+: [0,1]^d \to \mathbb{R}$ are extensions of $f$ defined as
\begin{equation*}
	f^- (s) = \max \left\{ g(s) \mid g(1_A) \leq f(A), \forall A\subset X,\ \  g: {\rm convex} \right\},
\end{equation*}
\begin{equation*}
	f^+ (s) = \min \left\{ g(s) \mid g(1_A) \geq f(A), \forall A\subset X,\ \  g: {\rm concave} \right\}.
\end{equation*}
\end{Definition}

From the definition, $f^-$ and $f^+$ are convex and concave, respectively.

\begin{Lemma}\label{lem_closure}
For a set function $f: 2^X \to \mathbb{R}$, its convex closure and concave closure are given by
\begin{equation*}
	f^- (s) = \min_{\alpha} \left\{ \left. \sum_{A \subset X} \alpha_A f(A) \right| \sum_{A \subset X} \alpha_A 1_A = s, \sum_{A \subset X} \alpha_A = 1, \alpha_A \geq 0 \right\},
\end{equation*}
\begin{equation}
	f^+ (s) = \max_{\alpha} \left\{ \left. \sum_{A \subset X} \alpha_A f(A) \right| \sum_{A \subset X} \alpha_A 1_A = s, \sum_{A \subset X} \alpha_A = 1, \alpha_A \geq 0 \right\}. \label{concave_closure}
\end{equation}
\end{Lemma}
\begin{proof}
See section 3.1.1 of \cite{Dughmi}.
\end{proof}

Another type of extension was introduced by \cite{Lovasz} and it plays an important role in submodular function optimization \citep{Fujishige}.

\begin{Definition}\citep{Lovasz}
For a set function $f: 2^X \to \mathbb{R}$, its \textit{Lov\'asz extension} $f^L: 
[0,1]^d \to \mathbb{R}$ is defined as
\begin{equation*}
	f^L(s) = \sum_{j=0}^d p_j(s) f (A_j(s)),
\end{equation*}
where $p_0(s),p_1(s),\cdots,p_d(s) \geq 0$ and $\emptyset=A_0(s) \subset A_1(s) \subset \cdots \subset A_d(s) = X$ are given by
\begin{equation*}
	s = \sum_{j=0}^d p_j(s) 1_{A_j(s)}, \quad \sum_{j=0}^d p_j(s) = 1, \quad |A_j(s)| = j.
\end{equation*}
\end{Definition}

Note that the Lov\'asz extension can be viewed as taking a special value of $\alpha$ in \eqref{concave_closure}.
\cite{Lovasz} showed that the submodularity of a set function $f$ is equivalent to the convexity of its Lov\'asz extension $f^L$.
In fact, there is a stronger result as follows.

\begin{Lemma}\label{lem_lovasz}
\begin{itemize}
\item The convex closure of a submodular function $f$ is equal to the Lov\'asz extension of $f$: $f^- = f^L$.

\item The concave closure of a supermodular function $f$ is equal to the Lov\'asz extension of $f$: $f^+ = f^L$.
\end{itemize}
\end{Lemma}
\begin{proof}
See section 3.1.3 of \cite{Dughmi}.
\end{proof}

\subsection{Relation between upper hedging price and concave closure}
Let $g: [ 0,1 ]^d \to [a^{(1)}_1,a^{(1)}_2] \times \cdots \times [a^{(d)}_1,a^{(d)}_2]$ be a bijective affine map defined by $g(s) = ((1-s_1) a^{(1)}_1+s_1 a^{(1)}_2, \cdots, (1-s_d) a^{(d)}_1+s_d a^{(d)}_2)$ and
$g_0: \{ 0,1 \}^d \to \chi$ be its restriction to $\{ 0,1 \}^d$.
By using $g_0$, we can identify the payoff function $f: \chi \to \mathbb{R}$ of the single-round game with a set function $f_0=f \circ g_0: \{ 0,1 \}^d \to \mathbb{R}$. 
Then, its concave closure is closely related to the maximization \eqref{single_upper} as follows.

\begin{Proposition}\label{lem_cc}
For a single-round game, the upper hedging price is given by
\begin{equation*}
	\bar{E}_{\chi} (f) = f_0^+ (g^{-1}(0)). 
\end{equation*}
\end{Proposition}
\begin{proof}
From Lemma \ref{lem_closure},
\begin{equation}
	f_0^+ (g^{-1}(0)) = \max_{\alpha} \left\{ \left. \sum_{A \subset X} \alpha_A f_0(1_A) \right| \sum_{A \subset X} \alpha_A 1_A = g^{-1}(0), \sum_{A \subset X} \alpha_A = 1, \alpha_A \geq 0 \right\}. \label{f0g0}
\end{equation}
Since $g$ is a bijective affine  map, the first constraint on $\alpha$ in \eqref{f0g0} is equivalent to
\begin{equation}
	\sum_{A \subset X} \alpha_A g(1_A) = 0. \label{alpha_cond}
\end{equation}
Here, for $A \subset X$, each entry of $g(1_A) \in \mathbb{R}^d$ is
\begin{equation*}
	g(1_A)_k = \begin{cases} a^{(k)}_2 & (x_k \in A), \\ a^{(k)}_1 & (x_k \not\in A). \end{cases}
\end{equation*}
Thus, \eqref{alpha_cond} is rewritten as
\begin{equation*}
	\sum_{A \subset X: x_k \in A} \alpha_A a^{(k)}_2 + \sum_{A \subset X: x_k \not\in A} \alpha_A a^{(k)}_1 = 0, \quad (k=1,\cdots,d).
\end{equation*}
Therefore, each $\alpha=(\alpha_A)_A$ satisfying the constraints in \eqref{f0g0} is viewed as a risk neutral measure on $\chi$.
Since $\bar{E}_{\chi} (f)$ is the maximum over all risk neutral measures, we obtain $\bar{E}_{\chi} (f) \geq f_0^+ (g^{-1}(0))$.

Conversely, for each $\widetilde{\chi} = \{ a_{i_0},\cdots,a_{i_d} \} \in \Gamma$, let
\begin{equation*}
	\alpha(A) = \begin{cases} p_j^{\widetilde{\chi}} & (g_0(1_A) = a_{i_j}) \\ 0 & ({\rm otherwise}) \end{cases},
\end{equation*}
where $p^{\widetilde{\chi}}$ is defined as \eqref{chi_prob}.
Then, 
\begin{equation*}
	\sum_{A \subset X} \alpha_A g(1_A) = \sum_{j=0}^d p_j^{\widetilde{\chi}} a_{i_j,k}= 0,
\end{equation*}
\begin{equation*}
	\sum_{A \subset X} \alpha_A = \sum_{j=0}^d p_j^{\widetilde{\chi}} = 1,
\end{equation*}
and
\begin{equation*}
	\sum_{A \subset X} \alpha_A f(A) = \sum_{j=0}^d p_j^{\widetilde{\chi}} f = I(\widetilde{\chi},f).
\end{equation*}
Therefore, from \eqref{f0g0},
\begin{equation*}
	f_0^+ (g^{-1}(0)) \geq I(\widetilde{\chi},f).
\end{equation*}
Since $\widetilde{\chi}$ is arbitrary, by Proposition \ref{prop_dual} we obtain
\begin{equation*}
	f_0^+ (g^{-1}(0)) \geq \max_{\widetilde{\chi} \in \Gamma} I(\widetilde{\chi},f) = \bar{E}_{\chi} (f).
\end{equation*}
\end{proof}

\subsection{Two assets case}
Suppose $d=2$ and $\chi = \{ a^{(1)}_1,a^{(1)}_2 \} \times \{ a^{(2)}_1,a^{(2)}_2 \}$, where $a^{(1)}_1 < a^{(1)}_2$ and $a^{(2)}_1 < a^{(2)}_2$.
%

For a single-round game, the maximization in \eqref{single_upper} involves two candidates of $\widetilde{\chi}$.
One of them ($\widetilde{\chi}_+$) has positive correlation while the other ($\widetilde{\chi}_-$) has negative correlation.
For example, in Figure \ref{LP_2d}, ABD and BCD have positive correlation while ABC and ACD have negative correlation.
Similarly to section 3.3, the payoff function $f: \chi \to \mathbb{R}$ is identified with a set function $f_0=f \circ g_0: \{ 0,1 \}^2 \to \mathbb{R}$. 
If $f_0$ is submodular or supermodular, then the maximizer in \eqref{single_upper} is determined as follows.

\begin{Proposition}\label{prop_2d}
\begin{itemize}
\item If $f_0$ is submodular, then the maximizer in \eqref{single_upper} is the one with negative correlation $\widetilde{\chi}_-$:
\begin{equation*}
	I(\widetilde{\chi}_-,f) = \max_{\widetilde{\chi} \in \Gamma} I(\widetilde{\chi},f).
\end{equation*}

\item If $f_0$ is supermodular, then the maximizer in \eqref{single_upper} is the one with positive correlation $\widetilde{\chi}_+$:
\begin{equation*}
	I(\widetilde{\chi}_+,f) = \max_{\widetilde{\chi} \in \Gamma} I(\widetilde{\chi},f).
\end{equation*}
\end{itemize}
\end{Proposition}
\begin{proof}
From the definition of the Lov\'asz extension,
\begin{equation*}
	f_0^L(g^{-1}(0)) = I(\widetilde{\chi}_+, f).
\end{equation*}

When $f_0$ is submodular, we have $f_0^-(g^{-1}(0)) = f_0^L(g^{-1}(0))$ from Proposition \ref{lem_lovasz}. Then, from Lemma \ref{lem_closure},
\begin{equation*}
	f_0^L(g^{-1}(0)) = I(\widetilde{\chi}_+,f) \leq I(\widetilde{\chi}_-,f)
\end{equation*}
and the maximizer in \eqref{single_upper} is $\widetilde{\chi}_-$.

When $f_0$ is supermodular, we have $f_0^+(g^{-1}(0)) = f_0^L(g^{-1}(0))$ from Proposition \ref{lem_lovasz}. Then, from Lemma \ref{lem_closure},
\begin{equation*}
	f_0^L(g^{-1}(0)) = I(\widetilde{\chi}_+,f) \geq I(\widetilde{\chi}_-,f)
\end{equation*}
and  the maximizer in \eqref{single_upper} is $\widetilde{\chi}_+$.
\end{proof}

Now, consider the $N$-round game and assume that the payoff function is a European option $f(\xi)=F(S_N)$, where $\xi=x_1 \cdots x_N$ and $S_N=x_1+\cdots+x_N$.
Recall that $\bar{f} (\xi^n,N-n)$ depends only on $S_n=x_1+\cdots+x_n$ for a European option.
The submodularity or supermodularity of the payoff function is preserved throughout backward induction of \eqref{N_upper} as follows.

\begin{Lemma}\label{lem_preserve}
\begin{itemize}
\item Suppose $F: \mathbb{R}^2 \to \mathbb{R}$ is submodular. 
Then, for $n=N-1,\cdots,0$, the composite function of $\bar{f}(\xi^n \cdot, N-n-1)$ and $g_0$ is submodular for every $\xi^n \in \chi^n$.

\item Suppose $F: \mathbb{R}^2 \to \mathbb{R}$ is supermodular. 
Then, for $n=N-1,\cdots,0$, the composite function of $\bar{f}(\xi^n \cdot, N-n-1)$ and $g_0$ is supermodular for every $\xi^n \in \chi^n$.
\end{itemize}
\end{Lemma}
\begin{proof}
We prove the first statement by induction on $n$.
The proof of the second statement is similar.

The case $n=N-1$ is trivial from $\bar{f}(\xi^N, 0) = F(S_N)$, $\bar{f} (\xi^{N-1} \cdot, 0) = f(\xi^{N-1} \cdot) = F(S_{N-1}+\cdot)$ and Definition \ref{def_conti}.

Now, assume that the composite function of $\bar{f}(\xi^n \cdot, N-n-1)$ and $g_0$ for every $\xi^n \in \chi^n$ is submodular as a set function.
Since $\bar{f} (\xi^n,N-n)$ depends only on $S_n=x_1+\cdots+x_n$, we write $\bar{f}(\xi^{n+1}, N-n-1) = \bar{F} (S_{n+1}, N-n-1)$.
Then, for every $\xi^{n-1} \in \chi^{n-1}$ and $A \subset X$, from \eqref{N_upper} and Theorem \ref{th_2d},
\begin{align*}
	\bar{f}(\xi^{n-1} g_0(A), N-(n-1)-1) &= \bar{f}(\xi^{n-1} g_0(A), N-n) \\
	&= \max_{\widetilde{\chi} \in \Gamma} I (\widetilde{\chi}, \bar{f}(\xi^{n-1} g_0(A) \cdot, N-n-1)) \\
	&= I (\widetilde{\chi}_-, \bar{f}(\xi^{n-1} g_0(A) \cdot, N-n-1)) \\
	&= \sum_{j=0}^d p_j^{\widetilde{\chi}_-} \bar{f}(\xi^{n-1} g_0(A) a_{i_j}, N-n-1) \\
	&= \sum_{j=0}^d p_j^{\widetilde{\chi}_-} \bar{F}(S_{n-1}+g_0(A)+a_{i_j}, N-n-1),
\end{align*}
where $\widetilde{\chi}_- = \{ a_{i_0},\cdots,a_{i_d} \}$.
Therefore,
\begin{align*}
	\bar{f}(\xi^{n-1} g_0(A), N-(n-1)-1) &= \sum_{j=0}^d p_j^{\widetilde{\chi}_-} \bar{F}(S_{n-1}+a_{i_j}+g_0(A), N-n-1).
\end{align*}
For every $j$, the function $A \mapsto \bar{F}(S_{n-1}+a_{i_j}+g_0(A), N-n-1)$ is the composition of $\bar{f}(\xi^{n-1} a_{i_j} \cdot, N-n-1)$ and $g_0$ and therefore it is submodular from assumption.
Then, since the submodularity is preserved under a convex combination, the composite function of $\bar{f}(\xi^{n-1} \cdot, N-(n-1)-1)$ and $g_0$ is also submodular.
\end{proof}

Combining Lemma \ref{lem_preserve} with Proposition \ref{prop_2d}, each maximization in \eqref{N_upper} is solved in closed form as follows.

\begin{Theorem}\label{th_2d}
\begin{itemize}
\item If $F: \mathbb{R}^2 \to \mathbb{R}$ is submodular, 
  then the maximizer in each step \eqref{N_upper} is the one with negative correlation $\widetilde{\chi}_-$.

\item If $F: \mathbb{R}^2 \to \mathbb{R}$ is supermodular, 
  then the maximizer in each step \eqref{N_upper} is the one with positive correlation $\widetilde{\chi}_+$.
\end{itemize}
\end{Theorem}

\subsection{The case of more assets}
Suppose $d \geq 3$ and $\chi = \{ a^{(1)}_1,a^{(1)}_2 \} \times \cdots \times \{ a^{(d)}_1,a^{(d)}_2 \}$, where $a^{(1)}_1 < a^{(1)}_2, \cdots, a^{(d)}_1 < a^{(d)}_2$.

Consider a single-round game.
Similarly to section 3.3, the payoff function $f: \chi \to \mathbb{R}$ is identified with a set function $f_0=f \circ g_0: \{ 0,1 \}^d \to \mathbb{R}$. 
If $f_0$ is supermodular, then as in the second part of 
Proposition \ref{prop_2d}, the maximizer in \eqref{single_upper} is determined as follows.

\begin{Proposition}\label{prop_hd}
If $f_0$ is supermodular, then the maximizer in \eqref{single_upper} is
\begin{equation}
  \widetilde{\chi}_L = \{ g (1_{A_0 (s)}), g(1_{A_1 (s)}), \cdots, g(1_{A_d (s)}) \},
  \ s=g^{-1}(0), \label{chiL}
\end{equation}
where $\emptyset=A_0(s) \subset A_1(s) \subset \cdots \subset A_d(s) = X$ are given by
\begin{equation*}
	s = \sum_{j=0}^d p_j 1_{A_j(s)}, \quad \sum_{j=0}^d p_j = 1, \quad |A_j(s)| = j.
\end{equation*}
\end{Proposition}

Now, consider the $N$-round game and assume that the payoff function is a European option $f(\xi)=F(S_N)$, where $\xi=x_1 \cdots x_N$ and $S_N=x_1+\cdots+x_N$.
Lemma \ref{lem_preserve} is extended to general $d$ by the same proof as follows.

\begin{Lemma}\label{lem_preserve2}
Suppose $F: \mathbb{R}^d \to \mathbb{R}$ is supermodular. 
Then, for $n=N-1,\cdots,0$, the composite function of $\bar{f}(\xi^n \cdot, N-n-1)$ and $g_0$ is supermodular as a set function for every $\xi^n \in \chi^n$.
\end{Lemma}

Combining Lemma \ref{lem_preserve2} and Proposition \ref{prop_hd}, each maximization in \eqref{N_upper} is solved in closed form as follows.

\begin{Theorem}\label{th_hd}
  If $F: \mathbb{R}^d \to \mathbb{R}$ is supermodular, 
  then the maximizer in each step \eqref{N_upper} is $\widetilde{\chi}_L$ in \eqref{chiL}.
\end{Theorem}

Thus, when $F$ is supermodular, the upper hedging price can be calculated efficiently.
On the other hand, when $F$ is submodular, the maximization in \eqref{single_upper} cannot be solved in closed form \citep{Calinescu}. 
Since the number of possible $\widetilde{\chi}$ grows at least as fast as $2^{d-2}$ (see Appendix B), 
the calculation of the upper hedging price of a European option with a submodular payoff function becomes intractable when $d$ is large.

\subsection{Options on the maximum or the minimum}
A typical and realistic example of multivariate contingent claims is the option on the maximum or the minimum of several assets \citep{Stulz,Johnson}.
The option on the maximum is a European option with the payoff function
\begin{equation*}
	f_M (\xi) = F_M (S_N) = (\max_i (S_N)_i-K)_{+},
\end{equation*}
where $x_{+} = \max (x,0)$.
Similarly, the option on the minimum is a European option with the payoff function
\begin{equation*}
	f_m (\xi) = F_m (S_N) = (\min_i (S_N)_i-K)_{+}.
\end{equation*}

\begin{Proposition}\label{prop_opt}
\begin{itemize}
\item The function $F_M$ is submodular. 

\item The function $F_m$ is supermodular. 
\end{itemize}
\end{Proposition}
\begin{proof}
Assume $F_M(s) \leq F_M(t)$ without loss of generality. 
Then $F_M(s \lor t) = F_M(t)$ since $\max_i (s \lor t)_i = \max_i t_i$.
Also, $F_M(s \land t) \leq F_M(s)$ since $\max_i (s \land t)_i \leq \max_i s_i$.
Therefore, we obtain $F_M(s \land t) + F_M(s \lor t) \leq F_M(s) + F_M(t)$.

Assume $F_m(s) \leq F_m(t)$ without loss of generality. 
Then $F_m(s \land t) = F_m(s)$ since $\min_i (s \land t)_i = \min_i s_i$.
Also, $F_m(s \lor t) \geq F_m(t)$ since $\min_i (s \lor t)_i \geq \min_i s_i$.
Therefore, we obtain $F_m(s \land t) + F_m(s \lor t) \geq F_m(s) + F_m(t)$.
\end{proof}

Combining Proposition \ref{prop_opt} with Theorem \ref{th_2d} and Theorem \ref{th_hd}, we obtain the following.

\begin{Corollary}\label{cor_opt}
Suppose $d=2$ and $\chi = \{ a^{(1)}_1,a^{(1)}_2 \} \times \{ a^{(2)}_1,a^{(2)}_2 \}$.

\begin{itemize}
\item For the option on the maximum, the maximizer in each step \eqref{N_upper} is the one with negative correlation $\widetilde{\chi}_-$.

\item For the option on the minimum, the maximizer in each step \eqref{N_upper} is the one with positive correlation $\widetilde{\chi}_+$.
\end{itemize}
\end{Corollary}

\begin{Corollary}
Suppose $d \geq 3$ and $\chi = \{ a^{(1)}_1,a^{(1)}_2 \} \times \cdots \times \{ a^{(d)}_1,a^{(d)}_2 \}$.
For the option on the minimum, the maximizer in each step \eqref{N_upper} is $\widetilde{\chi}_L$ in \eqref{chiL}.
\end{Corollary}


\section{Limiting behavior of upper hedging price of a European option}
\label{sec:limiting}
In this section, we show that the upper hedging price of a European option converges to the solution of the Black-Scholes-Barenblatt equation as the number of game rounds goes to infinity.
We also show that, when the payoff function is submodular or supermodular, 
the Black-Scholes-Barenblatt equation reduces to the linear Black-Scholes equation and it is solved in closed form.

\subsection{Derivation of the Black-Scholes-Barenblatt equation}
Consider an $N$-round multinomial game with a European option
\begin{equation}
	f (\xi^N) = F \left( \frac{S_N}{\sqrt{N}} \right), \quad S_N = x_1 + \cdots + x_N. \label{payoff}
\end{equation}

Recall $\Gamma=\{ \widetilde{\chi} \subset \chi \mid |\widetilde{\chi}| = d+1, \ 0 \in \conv \widetilde{\chi}, \ \dim \conv \widetilde{\chi}=d \}$.
For each $\widetilde{\chi} = \{ a_{l_0},\cdots,a_{l_d} \} \in \Gamma$, 
the risk neutral measure $p^{\widetilde{\chi}}=(p^{\widetilde{\chi}}_0,\cdots,p^{\widetilde{\chi}}_d)$ was defined as the solution of the linear equations \eqref{chi_prob}.
Let $\Sigma(\widetilde{\chi}) \in \mathbb{R}^{d \times d}$ be the covariance matrix of $p^{\widetilde{\chi}}$:
\begin{equation*}
	\left( \Sigma (\widetilde{\chi}) \right)_{ij} = \sum_{k=0}^d p^{\widetilde{\chi}}_k a_{l_k,i} a_{l_k,j} \quad (i,j=1,\cdots,d).
\end{equation*}
For a twice continuously differentiable $F$ 
we denote its Hessian matrix $\nabla_s^2 F (s) \in \mathbb{R}^{d \times d}$ as
\begin{equation*}
	\left(\nabla_s^2 F (s) \right)_{i j} = \frac{\partial^2}{\partial s_i \partial s_j} F (s) \quad (i,j=1,\cdots,d)
\end{equation*}

Then, by extending Theorem 4.1 of \cite{Nakajima}, we obtain the following.

\begin{Theorem}\label{th_BSB}
  Consider an $N$-round multinomial game with a European option \eqref{payoff}, where $F$ has a compact support and twice continuously differentiable. Assume  $0\not\in \chi$.

\begin{itemize}
\item The limit of the upper hedging price of a European option \eqref{payoff} is
\begin{equation*}
	\lim_{N \to \infty} \bar{E}_{\chi} (f) = \bar{u} (0,1),
\end{equation*}
where $\bar{u}: \mathbb{R}^d \times [0,1] \to \mathbb{R}$ is the solution of the partial differential equation
\begin{equation}
	\frac{\partial}{\partial t} \bar{u} (s,t) = \frac{1}{2} \max_{\widetilde{\chi} \in \Gamma} {\rm Tr} \left( \Sigma(\widetilde{\chi}) \cdot \nabla_s^2 \bar{u} (s,t) \right), \label{BSB}
\end{equation}
with the initial condition $\bar{u} (s,0) = F (s)$.

\item The limit of the lower hedging price of a European option \eqref{payoff} is
\begin{equation*}
	\lim_{N \to \infty} \underline{E}_{\chi} (f) = \underline{u} (0,1),
\end{equation*}
where $\underline{u}: \mathbb{R}^d \times [0,1] \to \mathbb{R}$ is the solution of the partial differential equation
\begin{equation}
	\frac{\partial}{\partial t} \underline{u} (s,t) = \frac{1}{2} \min_{\widetilde{\chi} \in \Gamma} {\rm Tr} \left( \Sigma(\widetilde{\chi}) \cdot \nabla_s^2 \underline{u} (s,t) \right), \label{BSB2}
\end{equation}
with the initial condition $\underline{u} (s,0) = F (s)$.
\end{itemize}
\end{Theorem}

By Theorem 4.6.9 of \cite{pham}
a smooth solution of \eqref{BSB} and \eqref{BSB2} exists
under the assumptions of Theorem \ref{th_BSB}.  The notion of viscosity solution is needed if
$0\in \chi$ or $F$ is only continuous.
As discussed in Section 6.3 of \cite{Shafer} the result of this theorem can be extended to the
case that the third and the fourth order derivatives of $F$ are bounded.
The PDE \eqref{BSB} is (the additive form of) the Black-Scholes-Barenblatt equation \citep{Romagnoli}.

When $d=1$, the maximum in \eqref{BSB} only depends on the sign of the second derivative of $\bar{u}$ and 
\eqref{BSB} reduces to the equation (13) in \cite{Nakajima}.
This case is also discussed in \cite{Peng} and Section 5 of \cite{Romagnoli}.
However, when $d \geq 2$, the maximization in \eqref{BSB} becomes more complicated.
This maximization is discussed from the viewpoint of optimization theory in Section 4 of \cite{Romagnoli}.

The Black-Scholes-Barenblatt equation in \cite{Romagnoli} is 
\begin{equation}
	\frac{\partial}{\partial t} \bar{u} (s,t) = \frac{1}{2} \max_{\lambda \in \Lambda} {\rm Tr} \left( \lambda \lambda^{\top} \cdot \nabla_s^2 \bar{u} (s,t) \right), \label{RomagnoliPDE}
\end{equation}
with the initial condition $\bar{u} (s,0) = F (s)$.
Here, $\Lambda$ is a closed bounded set in the space of $n \times n$ real matrices.
The PDE \eqref{BSB} is viewed as a special case of the PDE \eqref{RomagnoliPDE} where $\Lambda$ is a finite set.
In other words, the PDE \eqref{RomagnoliPDE} is a generalization of the PDE \eqref{BSB} to bounded forecasting games.
\cite{Nakajima} also discussed this point.


\subsection{Reduction to the linear Black-Scholes equation}
\cite{Romagnoli} stated that the Black-Scholes-Barenblatt equation \eqref{RomagnoliPDE} reduces to the ordinary linear Black-Scholes equation
\begin{equation}
	\frac{\partial}{\partial t} \bar{u} (s,t) = \frac{1}{2} {\rm Tr} \left( \lambda \lambda^{\top} \cdot \nabla_s^2 \bar{u} (s,t) \right), \label{BS}
\end{equation}
if the maximizing $\lambda$ does not depend on $s$ nor $t$.
Note that this PDE has the same form with the heat equation with anisotropic conductivity.
Such reduction occurs when the payoff function is submodular or supermodular and the move set $\chi$ is a product set \eqref{prod_chi} with $n_1=\cdots=n_d=2$.
Specifically, by taking $N \to \infty$ in Theorem \ref{th_2d} and \ref{th_hd}, we obtain the following.

\begin{Proposition}\label{prop_BS2}
Suppose $d=2$, $\chi = \{ a^{(1)}_1,a^{(1)}_2 \} \times \{ a^{(2)}_1,a^{(2)}_2 \}$ and $F: \mathbb{R}^2 \to \mathbb{R}$ is twice continuously differentiable.

\begin{itemize}
\item If $F$ is submodular, 
  then the PDE \eqref{BSB} is reduced to
\begin{equation*}
	\frac{\partial}{\partial t} \bar{u} (s,t) = \frac{1}{2} {\rm Tr} \left( \Sigma(\widetilde{\chi}_-) \cdot \nabla_s^2 \bar{u} (s,t) \right).
\end{equation*}

\item If $F$ is supermodular, 
  then the PDE \eqref{BSB} is reduced to
\begin{equation*}
	\frac{\partial}{\partial t} \bar{u} (s,t) = \frac{1}{2} {\rm Tr} \left( \Sigma(\widetilde{\chi}_+) \cdot \nabla_s^2 \bar{u} (s,t) \right).
\end{equation*}

\item If $F$ is submodular, 
  then the PDE \eqref{BSB2} is reduced to
\begin{equation*}
	\frac{\partial}{\partial t} \underline{u} (s,t) = \frac{1}{2} {\rm Tr} \left( \Sigma(\widetilde{\chi}_+) \cdot \nabla_s^2 \underline{u} (s,t) \right). 
\end{equation*}

\item If $F$ is supermodular, 
  then the PDE \eqref{BSB2} is reduced to
\begin{equation*}
	\frac{\partial}{\partial t} \underline{u} (s,t) = \frac{1}{2} {\rm Tr} \left( \Sigma(\widetilde{\chi}_-) \cdot \nabla_s^2 \underline{u} (s,t) \right).
\end{equation*}
\end{itemize}
\end{Proposition}

\begin{Proposition}\label{prop_BSh}
Suppose $d \geq 3$, $\chi = \{ a^{(1)}_1,a^{(1)}_2 \} \times \cdots \times \{ a^{(d)}_1,a^{(d)}_2 \}$ and $F: \mathbb{R}^d \to \mathbb{R}$ is twice continuously differentiable.
Define $\widetilde{\chi}_L$ by \eqref{chiL}.

\begin{itemize}
\item If $F$ is supermodular, 
  then the PDE \eqref{BSB} is reduced to
\begin{equation*}
	\frac{\partial}{\partial t} \bar{u} (s,t) = \frac{1}{2} {\rm Tr} \left( \Sigma(\widetilde{\chi}_L) \cdot \nabla_s^2 \bar{u} (s,t) \right).
\end{equation*}

\item If $F$ is submodular, 
  then the PDE \eqref{BSB2} is reduced to
\begin{equation*}
	\frac{\partial}{\partial t} \underline{u} (s,t) = \frac{1}{2} {\rm Tr} \left( \Sigma(\widetilde{\chi}_L) \cdot \nabla_s^2 \underline{u} (s,t) \right). 
\end{equation*}
\end{itemize}
\end{Proposition}

The linear Black-Scholes equation \eqref{BS} is solved in closed form by convolution of Gaussian densities.
Specifically, if $u_0(s)$ is smooth and has its third and fourth derivatives bounded, 
\begin{equation*}
	\frac{\partial}{\partial t} u = \frac{1}{2} {\rm Tr} \left( \Sigma \cdot \nabla_s^2 u \right)
\end{equation*}
with initial condition $u(s,t=0) = u_0(s)$ is given by
\begin{equation*}
	u(s,t) = \int \phi (s-s'; t \Sigma) u_0(s') {\rm d} s',
\end{equation*}
where $\phi (x; \Sigma)$ is the Gaussian density with covariance $\Sigma$:
\begin{equation*}
	\phi(x; \Sigma) = \frac{1}{(2 \pi)^{d/2} |\Sigma|^{1/2}} \exp \left( -\frac{1}{2} x^{\top} \Sigma^{-1} x \right).
\end{equation*}
Thus,
\begin{equation*}
	u (0,1) = \int F (s) \phi (s; \Sigma) {\rm d} s,
\end{equation*}
which is the expected value of $F(s)$ where $s$ has the distribution ${\rm N} (0,\Sigma)$.  Below
we write $s\sim {\rm N} (0,\Sigma)$.
Based on this, the upper and lower hedging prices are explicitly obtained as follows.

\begin{Theorem}\label{th_BS2}
Suppose $d=2$, $\chi = \{ a^{(1)}_1,a^{(1)}_2 \} \times \{ a^{(2)}_1,a^{(2)}_2 \}$ and $F: \mathbb{R}^2 \to \mathbb{R}$ is smooth and has its third and fourth derivatives bounded. 

\begin{itemize}
\item If $F$ is submodular, 
\begin{equation*}
  \lim_{N \to \infty} \bar{E}_{\chi} (f) = E [F(s)],
  \quad s \sim {\rm N} (0,\Sigma(\widetilde{\chi}_-)).
\end{equation*}

\item If $F$ is supermodular, 
\begin{equation*}
  \lim_{N \to \infty} \bar{E}_{\chi} (f) = E [F(s)],
  \quad s \sim {\rm N} (0,\Sigma(\widetilde{\chi}_+)).
\end{equation*}

\item If $F$ is submodular, 
\begin{equation*}
	\lim_{N \to \infty} \underline{E}_{\chi} (f) = E [F(s)],
        \quad s \sim {\rm N} (0,\Sigma(\widetilde{\chi}_+)).
\end{equation*}

\item If $F$ is supermodular, 
\begin{equation*}
  \lim_{N \to \infty} \underline{E}_{\chi} (f) = E [F(s)],
  \quad s \sim {\rm N} (0,\Sigma(\widetilde{\chi}_-)).
\end{equation*}
\end{itemize}
\end{Theorem}

\begin{Theorem}\label{th_BSh}
Suppose $d \geq 3$, $\chi = \{ a^{(1)}_1,a^{(1)}_2 \} \times \cdots \times \{ a^{(d)}_1,a^{(d)}_2 \}$ and $F: \mathbb{R}^d \to \mathbb{R}$ is smooth and has its third and fourth derivatives bounded. 
Define $\widetilde{\chi}_L$ by \eqref{chiL}.

\begin{itemize}
\item If $F$ is supermodular 
\begin{equation*}
  \lim_{N \to \infty} \bar{E}_{\chi} (f) = E [F(s)],
  \quad s \sim {\rm N} (0,\Sigma(\widetilde{\chi}_L)).
\end{equation*}

\item If $F$ is submodular 
  \begin{equation*}
  \lim_{N \to \infty} \underline{E}_{\chi} (f) = E [F(s)],
  \quad  s \sim {\rm N} (0,\Sigma(\widetilde{\chi}_L)).
\end{equation*}
\end{itemize}
\end{Theorem}

In Theorem \ref{th_BS2} and \ref{th_BSh} we assumed that
$F$ is smooth and has its third and fourth derivatives bounded. 
However Theorem \ref{th_BS2} and \ref{th_BSh} hold for non-smooth payoff functions, such as options on 
the maximum or the minimum, because these payoff functions can be uniformly approximated by
smooth payoff functions.


\section{Numerical results}
In this section, we confirm the validity of the theoretical results by numerical experiments.

To calculate the asymptotic value of the upper and lower hedging prices, we solve the Black-Scholes-Barenblatt equation 
\begin{equation*}
	\frac{\partial}{\partial t} {u} (s,t) = \frac{1}{2} \max_{\widetilde{\chi} \in \Gamma} {\rm Tr} \left( \Sigma(\widetilde{\chi}) \cdot \nabla_s^2 {u} (s,t) \right)
\end{equation*}
with the initial condition ${u} (s,0) = u_0 (s)$ by the finite difference method \citep{Smith}.
We explain the method for $d=2$ in the following. 
Let $\Delta s$ and $\Delta t = 1/K$ be the step sizes in space and time, respectively.
Here, we use the same value $\Delta s$ of step size in each space dimension for simplicity and $K>0$ is an integer.
We restrict the domain of $s$ to the rectangle $D=[-M \Delta s, M \Delta s] \times [-M \Delta s, M \Delta s]$, where $M>0$ is a sufficiently large integer.
Let $u_{i,j,k}$ be the numerical value of $u(i \Delta s, j \Delta s, k \Delta t)$ for $i=-M,-M+1,\cdots,M-1,M$, $j=-M,-M+1,\cdots,M-1,M$ and $k=0,1,\cdots,K$.
From the initial condition, we set $u_{i,j,0}=u_0(i \Delta s, j \Delta s)$.
To calculate $u_{i,j,k}$ for $k=1,2,\cdots,K$ iteratively,
we employ the explicit Euler scheme and select $\widetilde{\chi}$ at each step by comparing all elements in $\Gamma$,
following the approach of \cite{Nakajima} for $d=1$.
Specifically, the scheme is written as
\begin{equation}
	u_{i,j,k+1} = u_{i,j,k} + \frac{1}{2} \max_{\widetilde{\chi} \in \Gamma} {\rm Tr} \left( \Sigma(\widetilde{\chi}) \cdot \nabla_{{\rm d}}^2 u_{i,j,k} \right) \Delta t, \label{scheme}
\end{equation}
where $\nabla_{{\rm d}}^2 u_{i,j,k} \in \mathbb{R}^{2 \times 2}$ is the discretization of the Hessian defined as
\begin{equation*}
	\left( \nabla_{{\rm d}}^2 u_{i,j,k} \right)_{11} = \frac{u_{i+1,j,k}-2 u_{i,j,k}+u_{i-1,j,k}}{(\Delta s)^2},
\end{equation*}
\begin{equation*}
	\left( \nabla_{{\rm d}}^2 u_{i,j,k} \right)_{12} = \left( \nabla_{{\rm d}}^2 u_{i,j,k} \right)_{21} = \frac{u_{i+1,j+1,k}-u_{i+1,j-1,k}-u_{i-1,j+1,k}+u_{i-1,j-1,k}}{4 (\Delta s)^2},
\end{equation*}
\begin{equation*}
	\left( \nabla_{{\rm d}}^2 u_{i,j,k} \right)_{22} = \frac{u_{i,j+1,k}-2 u_{i,j,k}+u_{i,j-1,k}}{(\Delta s)^2}.
\end{equation*}
We solve \eqref{scheme} iteratively for $k=0,1,\cdots,K-1$ with the Dirichlet boundary condition $u_{i,j,k} = u_0(i \Delta s, j \Delta s)$ 
where $(i \Delta s,j \Delta s) \in \partial D$.
Then, we adopt $u_{0,0,K} \approx u(0,0,1)$ as the approximate value of the upper hedging price.
To avoid numerical instability, the step sizes must satisfy
\begin{equation*}
	\frac{\Delta t}{(\Delta s)^2} \leq \frac{1}{2},
\end{equation*}
which is called the Crank-Nicolson condition \citep{Smith}. 

\subsection{Options on the maximum or the minimum of two assets}
Here, we calculate the upper and lower hedging prices of options on the maximum or the minimum of two assets \citep{Stulz}.
Specifically, the option on the maximum of two assets is defined as
\begin{equation}
    f_M (\xi) = F_M (S_N) = (\max ((S_N)_1, (S_N)_2) - K)_{+} \label{max_option_2d}
\end{equation}
and the option on the minimum of two assets is defined as
\begin{equation}
    f_m (\xi) = F_m (S_N) = (\min ((S_N)_1, (S_N)_2) - K)_{+}. \label{min_option_2d}
\end{equation}
In our experiments we take $K=1$.

Suppose the move set is $\chi_1 = \{ (1,1), (1,-1), (-1,1), (-1,-1) \} = \{ -1, 1 \} \times \{ -1, 1 \}$.
Figure \ref{cross_price} plots the upper and lower hedging prices calculated by solving \eqref{N_upper} recursively for each $N$. 
The calculated prices almost converge around $N=5$. 
In this case, from Theorem \ref{th_BS2}, the limits of upper and lower hedging prices are obtained in closed form. 
Note that
\begin{equation*}
	\Sigma(\widetilde{\chi}_+) = \begin{pmatrix} 1 & 1 \\ 1 & 1 \end{pmatrix}, \quad \Sigma(\widetilde{\chi}_-) = \begin{pmatrix} 1 & -1 \\ -1 & 1 \end{pmatrix}.
\end{equation*}
Thus, $s \sim {\rm N} (0,\Sigma(\widetilde{\chi}_+))$ means $s_1=s_2 \sim {\rm N}(0,1)$ while $s \sim {\rm N} (0,\Sigma(\widetilde{\chi}_-))$ means $s_1=-s_2 \sim {\rm N}(0,1)$.
Therefore, the limits of the upper hedging prices are
\begin{equation*}
	\bar{E}_{\chi_1} (f_M) = \int_{-\infty}^{\infty} (|z|-1)_+ \phi (z) {\rm d} z = 2 \int_1^{\infty} (z-1)_+ \phi (z) {\rm d} z = 2(\phi(1)-\Phi(-1)) = 0.1666,
\end{equation*}
\begin{equation*}
	\bar{E}_{\chi_1} (f_m) = \int_{-\infty}^{\infty} (z-1)_+ \phi (z) {\rm d} z = \int_1^{\infty} (z-1)_+ \phi (z) {\rm d} z = \phi(1)-\Phi(-1) = 0.0833,
\end{equation*}
and the limits of the lower hedging prices are
\begin{equation*}
	\underline{E}_{\chi_1} (f_M) = \int_{-\infty}^{\infty} (z-1)_+ \phi (z) {\rm d} z = 0.0833,
\end{equation*}
\begin{equation*}
	\underline{E}_{\chi_1} (f_m) = \int_{-\infty}^{\infty} (-|z|-1)_+ \phi (z) {\rm d} z = 0.
\end{equation*}
These values are shown in Figure \ref{cross_price} by the horizontal lines.
They agree well with the convergence values.

\begin{figure}
\begin{minipage}{0.45\textwidth}
(a)
\begin{center}
 \includegraphics[width=6cm]{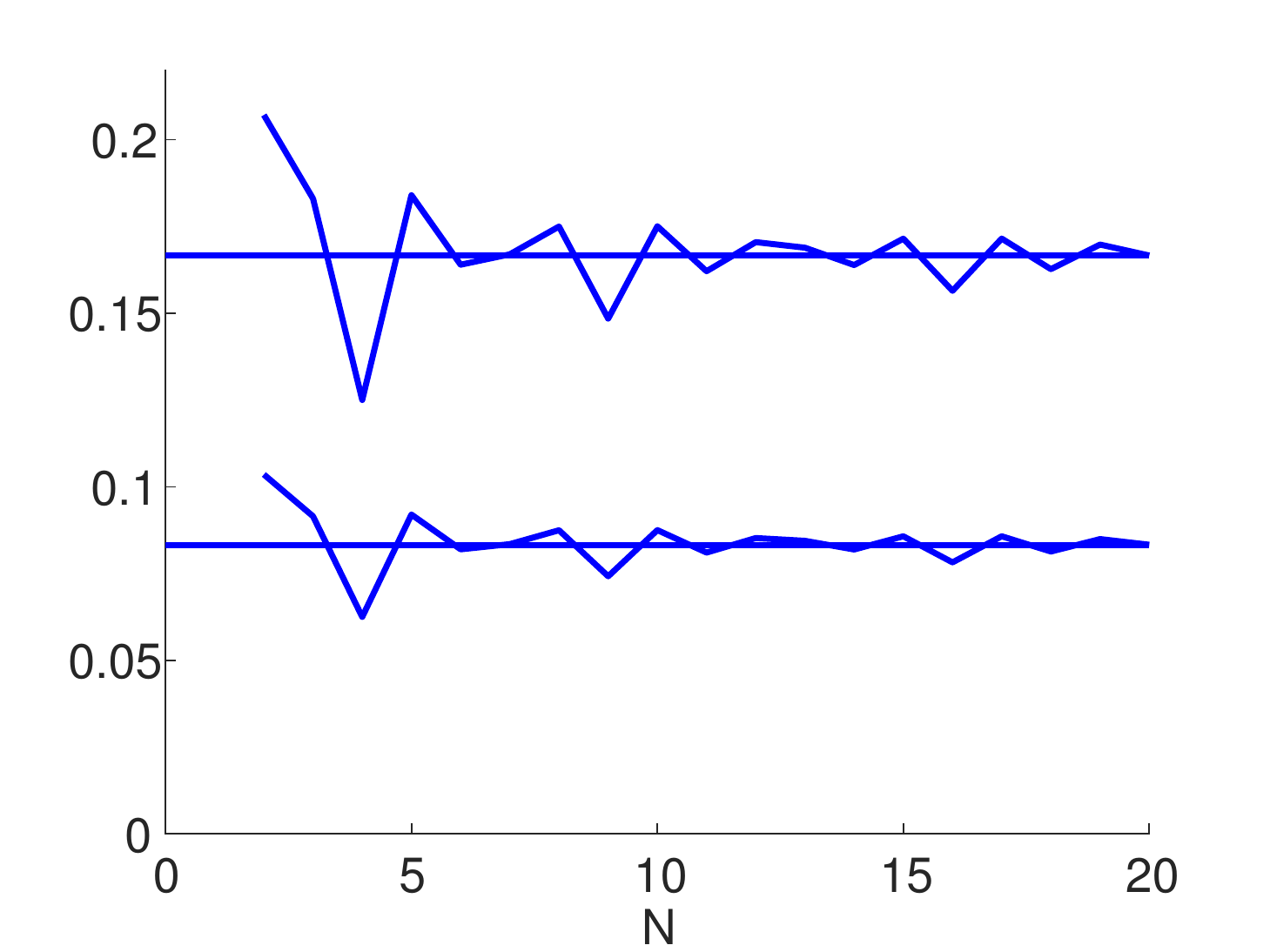}
\end{center}
\end{minipage}
\begin{minipage}{0.45\textwidth}
(b)
\begin{center}
 \includegraphics[width=6cm]{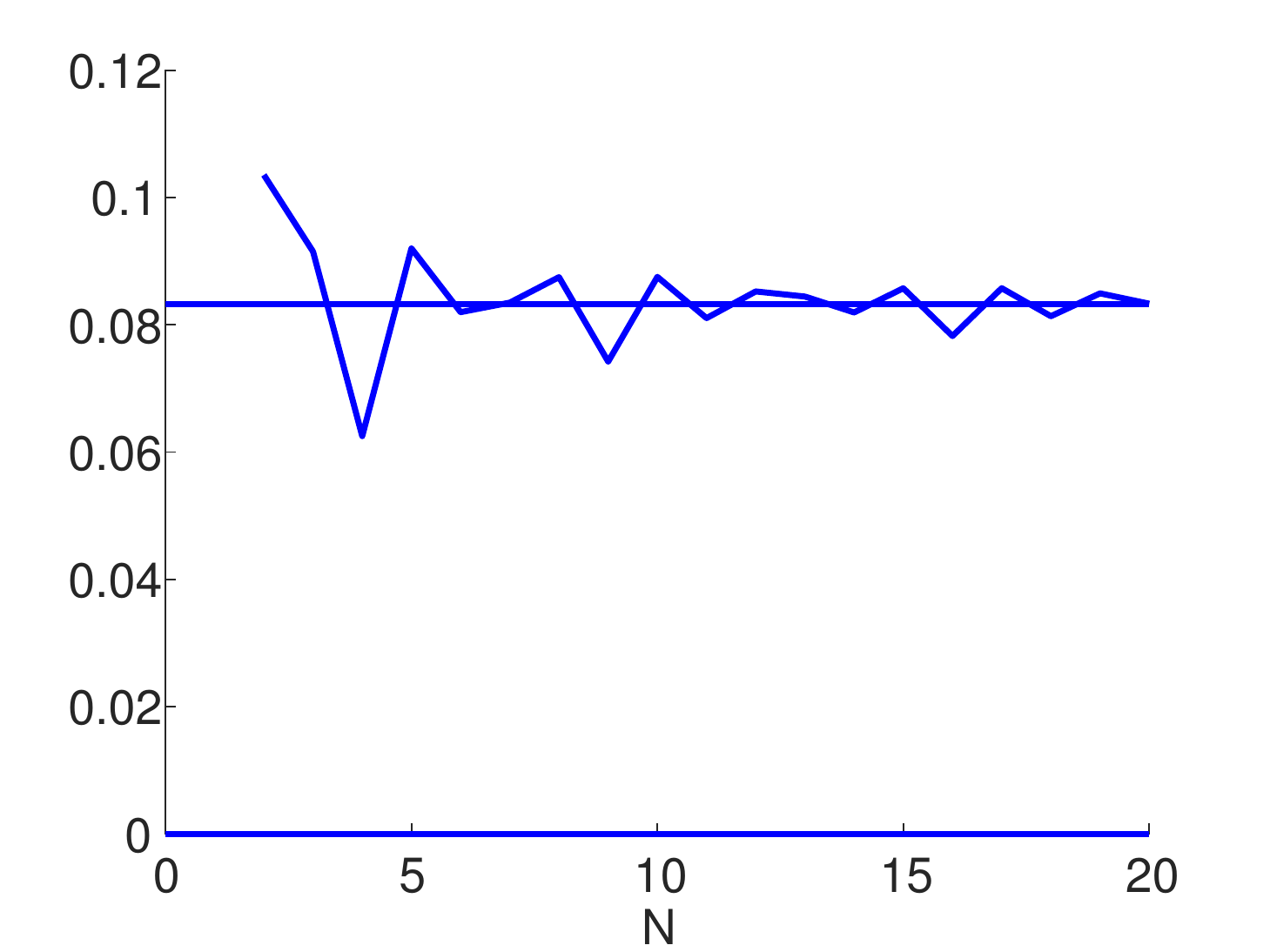}
\end{center}
\end{minipage}
\caption{Upper and lower hedging prices with move set $\chi_1$. (a) Option on the maximum \eqref{max_option_2d}. (b) Option on the minimum  \eqref{min_option_2d}. The horizontal lines represent the limit values. In (b), the lower hedging prices calculated for each $N$ are almost zero}
\label{cross_price}
\end{figure}

The move set $\chi_1$ corresponds to a lattice binomial model considered in \cite{Boyle}.
Based on this model, \cite{Boyle} proposed a method for pricing multivariate contingent claims by specifying the correlation coefficient between two assets.
Specifically, if the correlation coefficient $\rho$ between two assets under a risk neutral measure $p$ is given, then $p$ is uniquely determined as
\begin{equation*}
	p = \begin{pmatrix} 1 & 1 & 1 & 1\\ 1 & 1 & -1 & -1 \\ 1 & -1 & 1 & -1 \\ 1 & -1 & -1 & 1 \end{pmatrix}^{-1} \begin{pmatrix} 1 \\ 0 \\ 0 \\ \rho \end{pmatrix}.
\end{equation*}
Then, we can use the Cox-Ross-Rubinstein formula \citep{Cox} to calculate the prices.
Namely, we take the expectation with respect to $p$:
\begin{equation}
	E_p (f) = \sum_{\xi \in \chi^N} p_{x_1} \cdots p_{x_N} f(\xi), \quad \xi=x_1 \cdots x_N. \label{Boyle_price}
\end{equation}
Note that this method does not provide the upper or lower hedging price in general.
Figure \ref{maxmin_corr} plots the prices calculated by \eqref{Boyle_price} for each value of $\rho$ where $N=20$.
For the option on the maximum \eqref{max_option_2d}, the prices for $\rho=-1$ and $\rho=1$ coincide with the upper and lower hedging prices, respectively.
For the option on the minimum \eqref{min_option_2d}, the prices for $\rho=1$ and $\rho=-1$ coincide with the upper and lower hedging prices, respectively.
These are understood from Corollary \ref{cor_opt}, 
because $\rho=1$ means that $\widetilde{\chi}_+$ is always taken while $\rho=-1$ means that $\widetilde{\chi}_-$ is always taken.
When $d \geq 3$, the risk neutral measure is not uniquely determined by specifying correlation coefficients because $1+d+d(d-1)/2 < 2^d$.
Although \cite{Boyle} select one risk neutral measure by considering symmetry, there seems to be no theoretical support of this selection.

\begin{figure}
\begin{minipage}{0.45\textwidth}
(a)
\begin{center}
 \includegraphics[width=6cm]{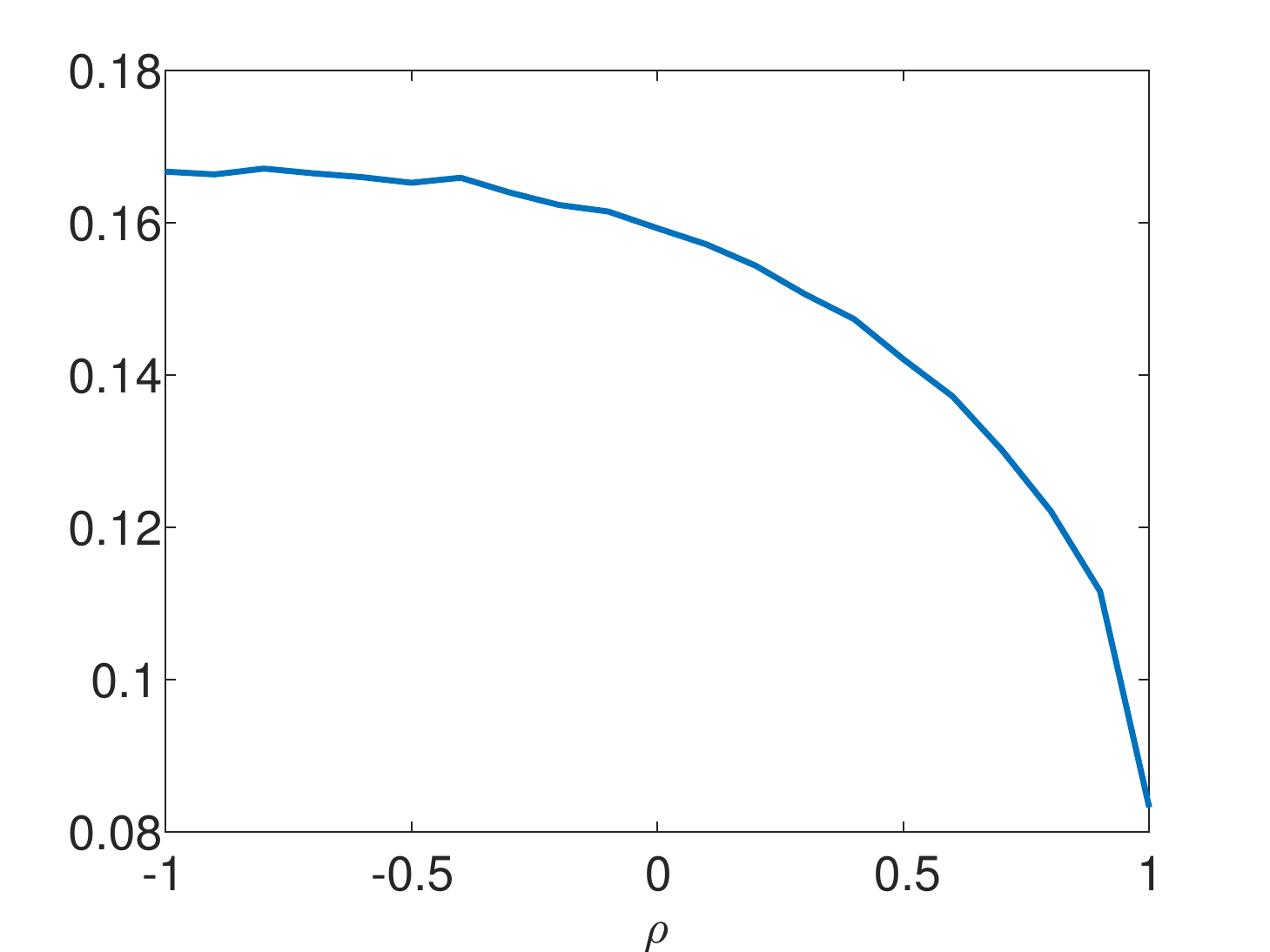}
\end{center}
\end{minipage}
\begin{minipage}{0.45\textwidth}
(b)
\begin{center}
 \includegraphics[width=6cm]{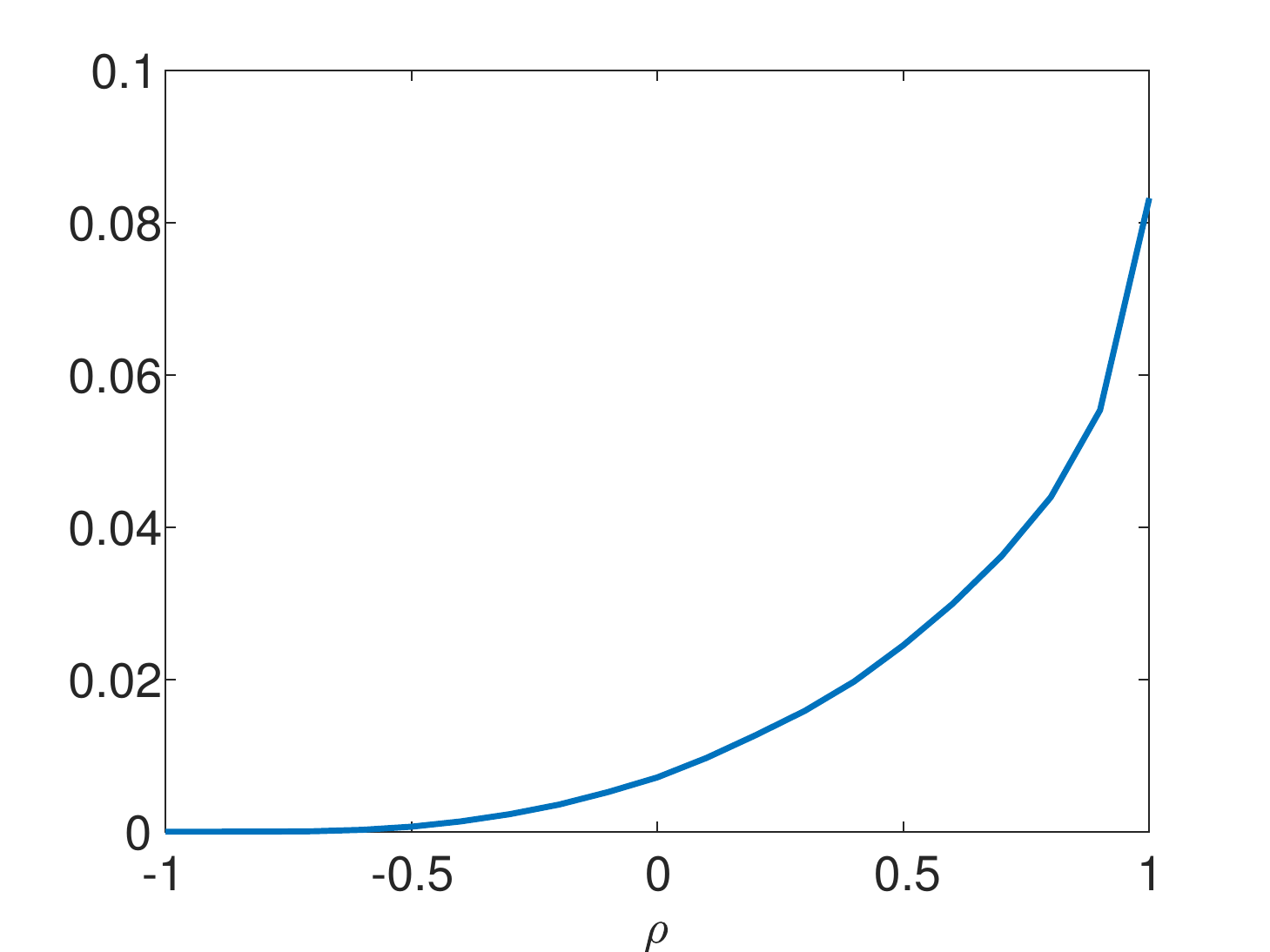}
\end{center}
\end{minipage}
	\caption{Prices calculated by the method of \cite{Boyle} for each value of $\rho$. (a) Option on the maximum \eqref{max_option_2d}. (b) option on the minimum \eqref{min_option_2d}.}
	\label{maxmin_corr}
\end{figure}

Now, suppose the move set is $\chi_2 =\{ (1,0), (-1,0), (0,1), (0,-1) \}$.
Figure \ref{plus_price} plots the upper and lower hedging prices calculated by solving \eqref{N_upper} recursively for each $N$. 
We also calculated the limits of upper and lower hedging prices by solving the Black-Scholes-Barenblatt equation \eqref{BSB} numerically. 
In the finite difference method, we set the step sizes to $\Delta s = 1/10$ and $\Delta t = 1/300$, which satisfy the Crank-Nicolson condition,
and restricted the domain of $s$ to $D = [-7,7] \times [-7,7]$.
The calculated values are
\begin{equation}
	\bar{E}_{\chi_2} (f_M) \approx 0.1105, \underline{E}_{\chi_2} (f_M) \approx 0.0084, \bar{E}_{\chi_2} (f_m) \approx 0.0028, \underline{E}_{\chi_2} (f_m) \approx 0. \label{limit_plus}
\end{equation}
These values are shown in Figure \ref{plus_price} by the horizontal lines.

\begin{figure}
\begin{minipage}{0.45\textwidth}
(a)
\begin{center}
 \includegraphics[width=6cm]{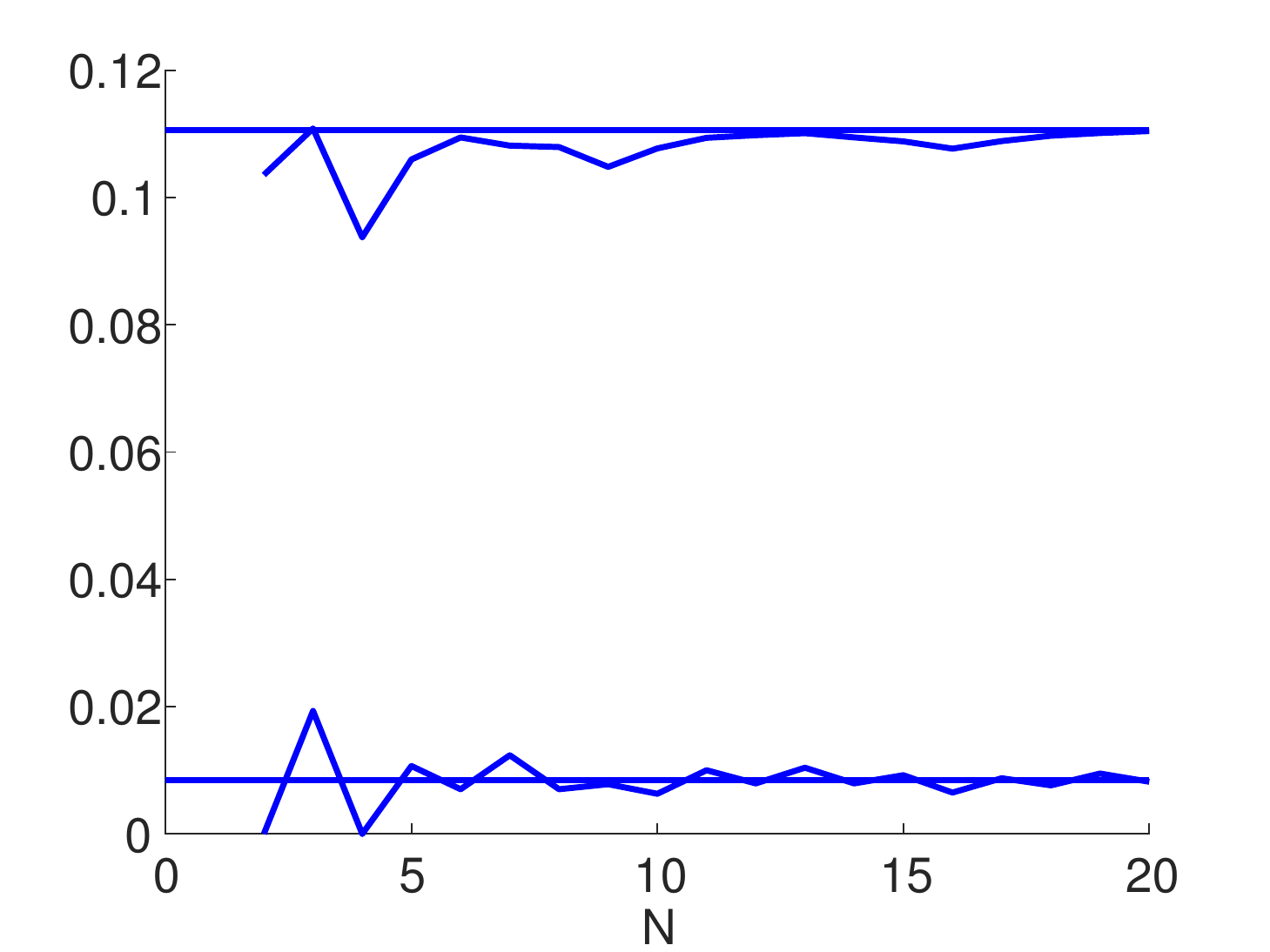}
\end{center}
\end{minipage}
\begin{minipage}{0.45\textwidth}
(b)
\begin{center}
 \includegraphics[width=6cm]{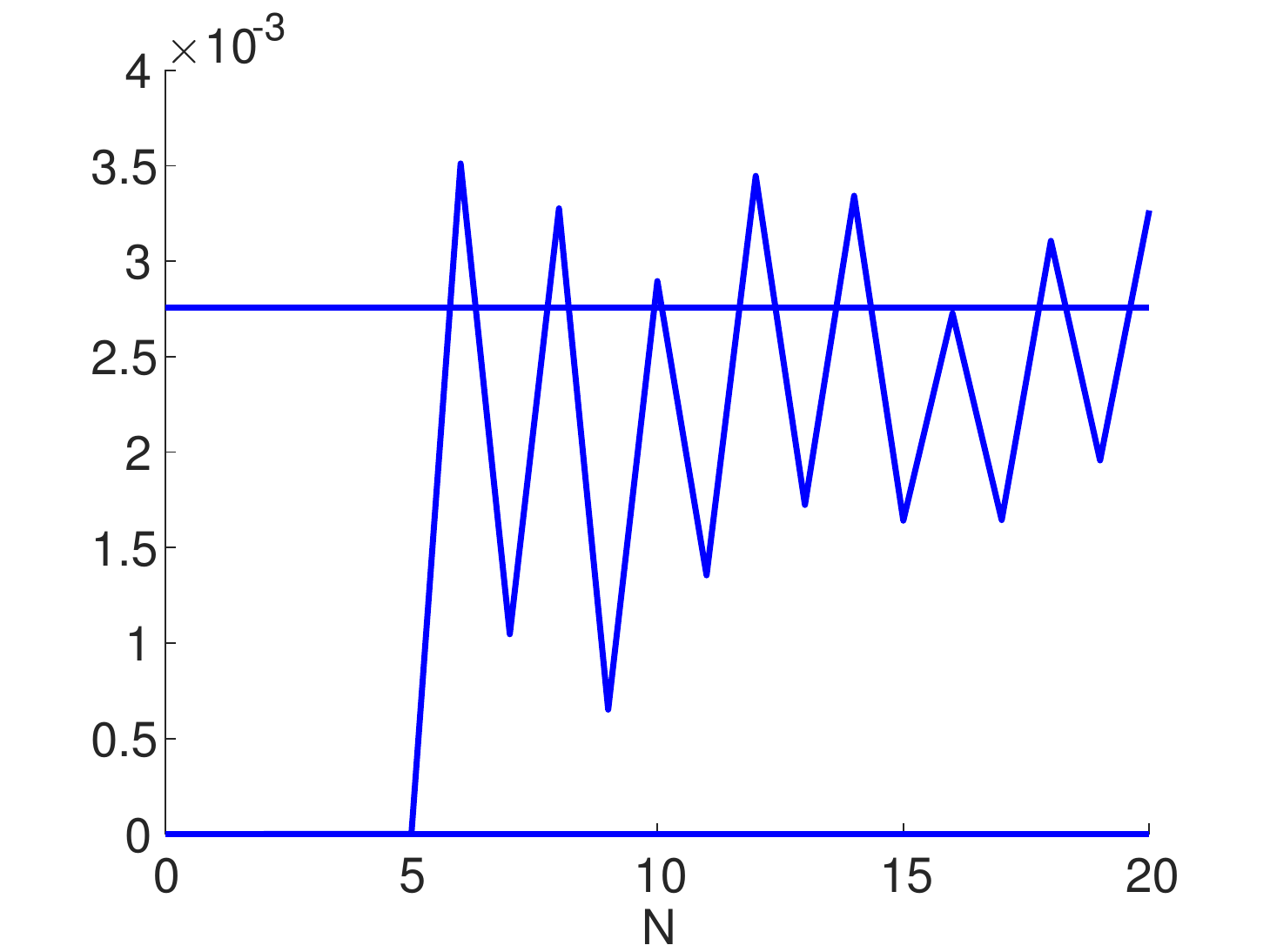}
\end{center}
\end{minipage}
\caption{Upper and lower hedging prices with move set $\chi_2$. (a) Option on the maximum \eqref{max_option_2d}. (b) Option on the minimum \eqref{min_option_2d}. The horizontal lines represent the limit values. In (b), the lower hedging prices calculated for each $N$ are almost zero}
\label{plus_price}
\end{figure}

\subsection{Option on the minimum of three assets}
Here, we calculate the upper hedging price of the option on the minimum of three assets \citep{Johnson} defined as
\begin{equation}
    f_m (\xi) = F_m(S_N) = (\min ((S_N)_1,(S_N)_2,(S_N)_3) - K)_{+}. \label{min_option_3d_f}
\end{equation}
In our experiments we take $K=1$.

Suppose the move set is $\chi = \{ -1, 2 \} \times \{ -2, 1 \} \times \{ -1, 1 \}$.
Then, $\widetilde{\chi}_L$ in \eqref{chiL} and $\Sigma(\widetilde{\chi}_L)$ are obtained as
\begin{equation}
	\widetilde{\chi}_L = \{ (-1,-2,-1),(-1,1,-1),(-1,1,1),(2,1,1) \},
\end{equation}
\begin{equation}
	\Sigma(\widetilde{\chi}_L) = \begin{pmatrix} 2 & 1 & 1 \\ 1 & 2 & 1 \\ 1 & 1 & 1 \end{pmatrix}.
\end{equation}
Figure \ref{min_option_3d} plots the upper hedging price calculated by solving \eqref{N_upper} recursively.
It almost converges around $N=5$.
On the other hand, from Theorem \ref{th_BSh},
\begin{equation*}
	\bar{E}_{\chi_1} (f_m) = E[F_m(s)],
\end{equation*}
where $s \in {\rm N}(0,\Sigma(\widetilde{\chi}_L))$.
We calculated this expectation by the Monte Carlo method with $10^6$ samples and obtained
\begin{equation}
	E[F_m(s)] \approx 0.0374. \label{limit_3d}
\end{equation}
This value is shown in Figure \ref{min_option_3d} by the horizontal line.
It agrees well with the convergence value.

\begin{figure}
\begin{center}
 \includegraphics[width=6cm]{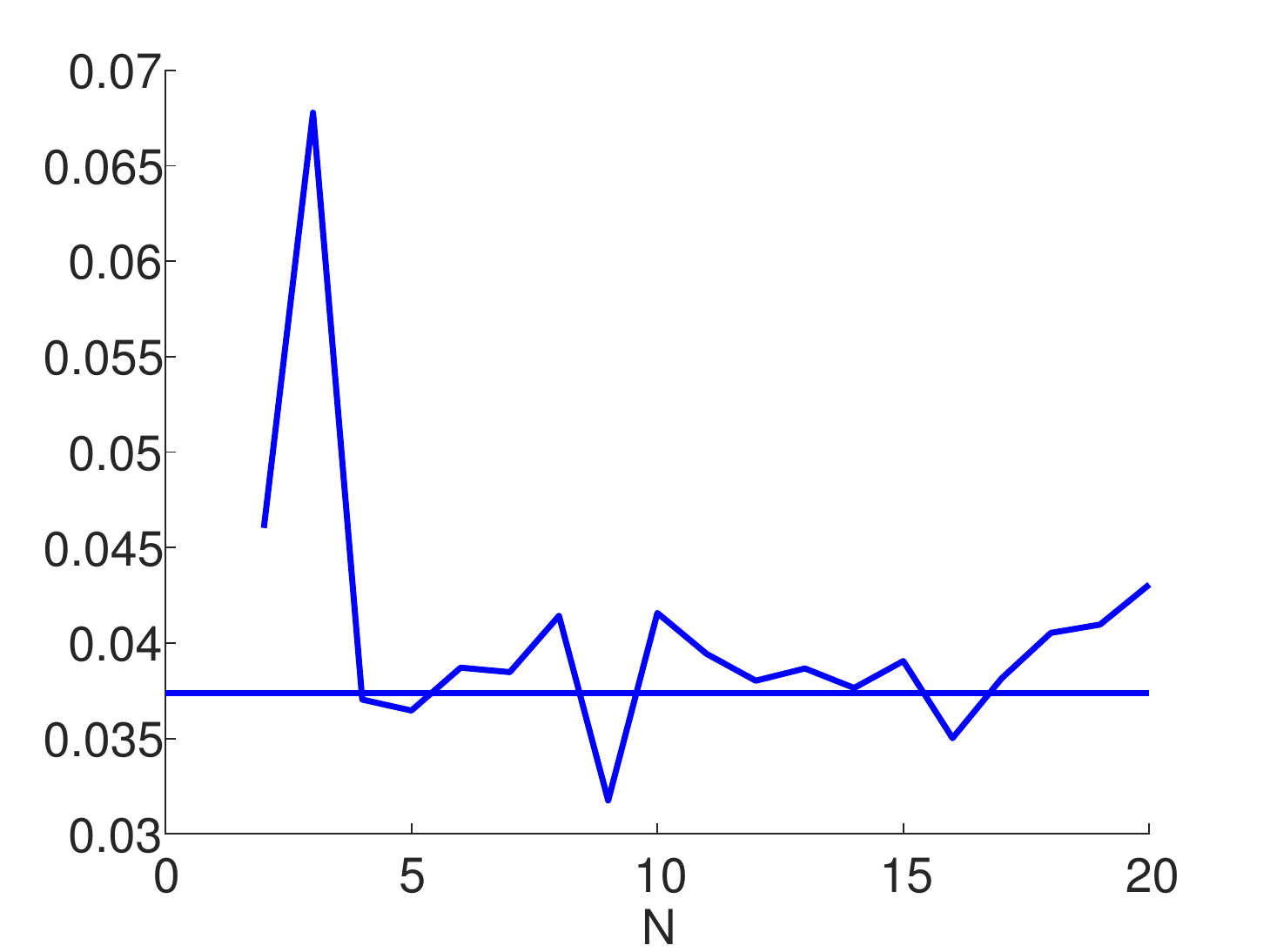}
\end{center}
\caption{Upper hedging price of the option on the minimum \eqref{min_option_3d_f} with move set $\chi = \{ -1, 2 \} \times \{ -2, 1 \} \times \{ -1, 1 \}$. The horizontal line represents the limit values}
\label{min_option_3d}
\end{figure}

\subsection{Butterfly-type options}
Finally, we consider European options motivated from Butterfly spread options.
\cite{Nakajima} calculated the upper and lower hedging prices of the Butterfly spread options when $d=1$: 
\begin{equation*}
	f(\xi)=F (S_N)=\max (0, S_N+0.5) - 2 \max (0, S_N-0.5) + \max (0, S_N-1.5).
\end{equation*}
Here, we consider the case where $d=2$ and the move set is $\chi_1 = \{ -1, 1 \} \times \{ -1, 1 \}$ or $\chi_2 = \{ (1,0), (-1,0), (0,1), (0,-1) \}$.
In numerical solution of the Black-Scholes-Barenblatt equation \eqref{BSB} by the finite difference method, 
we set the step sizes to $\Delta s = 1/10$ and $\Delta t = 1/300$, which satisfy the Crank-Nicolson condition,
and restricted the domain of $s$ to $D = [-7,7] \times [-7,7]$.

\subsubsection{Non-separable case}
Consider a European option defined as
\begin{equation}
    f(\xi) = F (S_N) = \begin{cases}
		0 & ((S_N)_1 < -0.5+|(S_N)_2|) \\
		(S_N)_1-(-0.5+|(S_N)_2|) & (-0.5+|(S_N)_2| \leq (S_N)_1 < 0.5) \\
		(1.5-|(S_N)_2|)-(S_N)_1 & (0.5 \leq (S_N)_1 < 1.5-|(S_N)_2|) \\
		0 & (1.5-|(S_N)_2| \leq (S_N)_1) \label{cone_f}
	\end{cases}.
\end{equation}
When $(S_N)_2$ is fixed, this payoff function behaves like the Butterfly spread option as a function of $(S_N)_1$.

Figure \ref{cone_price} plots the upper and lower hedging prices calculated by solving \eqref{N_upper} recursively for each $N$.
On the other hand, the limit values calculated by solving the Black-Scholes-Barenblatt equation \eqref{BSB} numerically are
\begin{equation*}
	\bar{E}_{\chi_1} (f) \approx 0.1786, \underline{E}_{\chi_1} (f) \approx 0.0028, \bar{E}_{\chi_2} (f) \approx 0.3315, \underline{E}_{\chi_2} (f) \approx 0.0470.
\end{equation*}
These values are shown in Figure \ref{cone_price} by the horizontal lines.
They agree well with the convergence values.

\begin{figure}
\begin{minipage}{0.45\textwidth}
(a)
\begin{center}
 \includegraphics[scale=0.3]{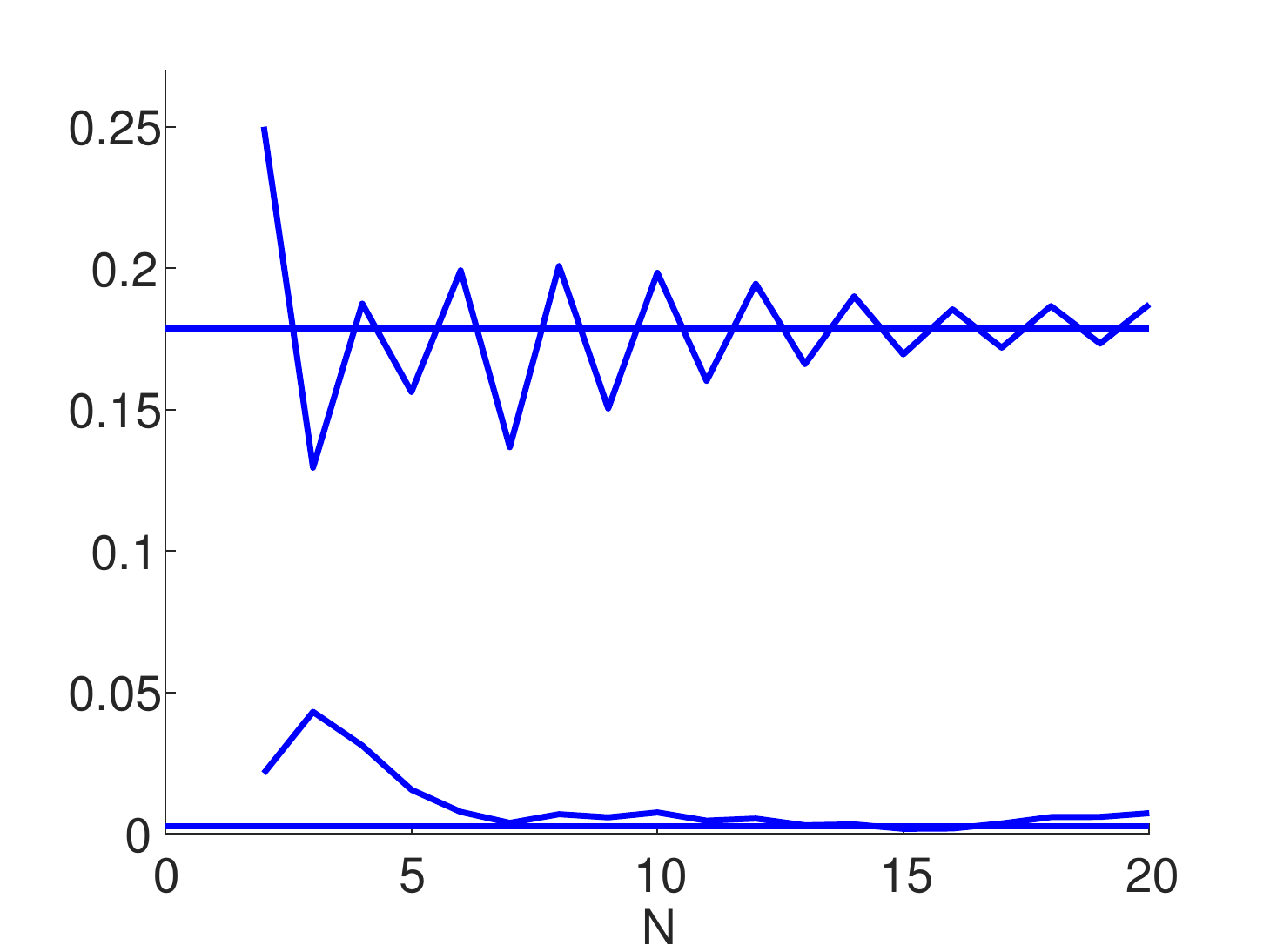}
\end{center}
\end{minipage}
\begin{minipage}{0.45\textwidth}
(b)
\begin{center}
 \includegraphics[scale=0.3]{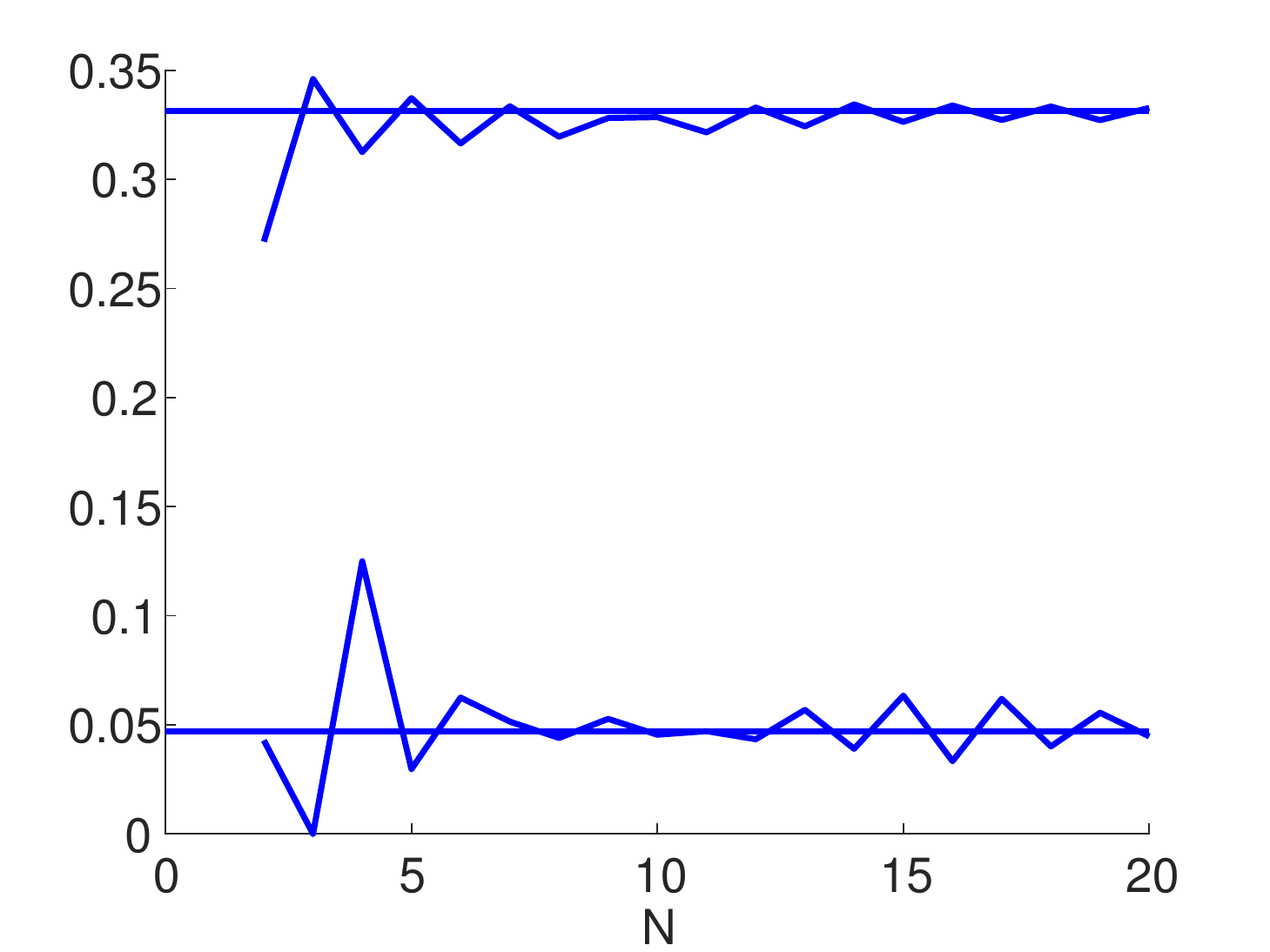}
\end{center}
\end{minipage}
	\caption{Upper and lower hedging prices of the European option \eqref{cone_f} with move set (a) $\chi _1$ and (b) $\chi_2$. Horizontal lines represent the limit values}
	\label{cone_price}
\end{figure}

\subsubsection{Separable case}
Consider a European option defined as
\begin{equation}
    f (\xi) = F (S_N) = g((S_N)_1) + g((S_N)_2), \label{double_butterfly_f}
\end{equation}
where
\begin{equation}
	g(s) = \max (0, s+0.5) - 2 \max (0, s-0.5) + \max (0, s-1.5)
\end{equation}
is the butterfly spread option used in \cite{Nakajima}.
This option is separable.

Figure \ref{double_butterfly_price} plots the upper and lower hedging prices calculated by solving \eqref{N_upper} recursively for each $N$.
On the other hand, the limit values calculated by solving the Black-Scholes-Barenblatt equation \eqref{BSB} numerically are
\begin{equation*}
	\bar{E}_{\chi_1} (f) \approx 0.6609, \underline{E}_{\chi_1} (f) \approx 0.6609, \bar{E}_{\chi_2} (f) \approx 1.0938, \underline{E}_{\chi_2} (f) \approx 0.5640.
\end{equation*}
These values are shown in Figure \ref{double_butterfly_price} by the horizontal lines.
They agree well with the convergence values.
Note that the upper and lower hedging prices coincide for $\chi_1$, which illustrates Proposition \ref{prop_sep}.

\begin{figure}
\begin{minipage}{0.45\textwidth}
(a)
\begin{center}
 \includegraphics[scale=0.3]{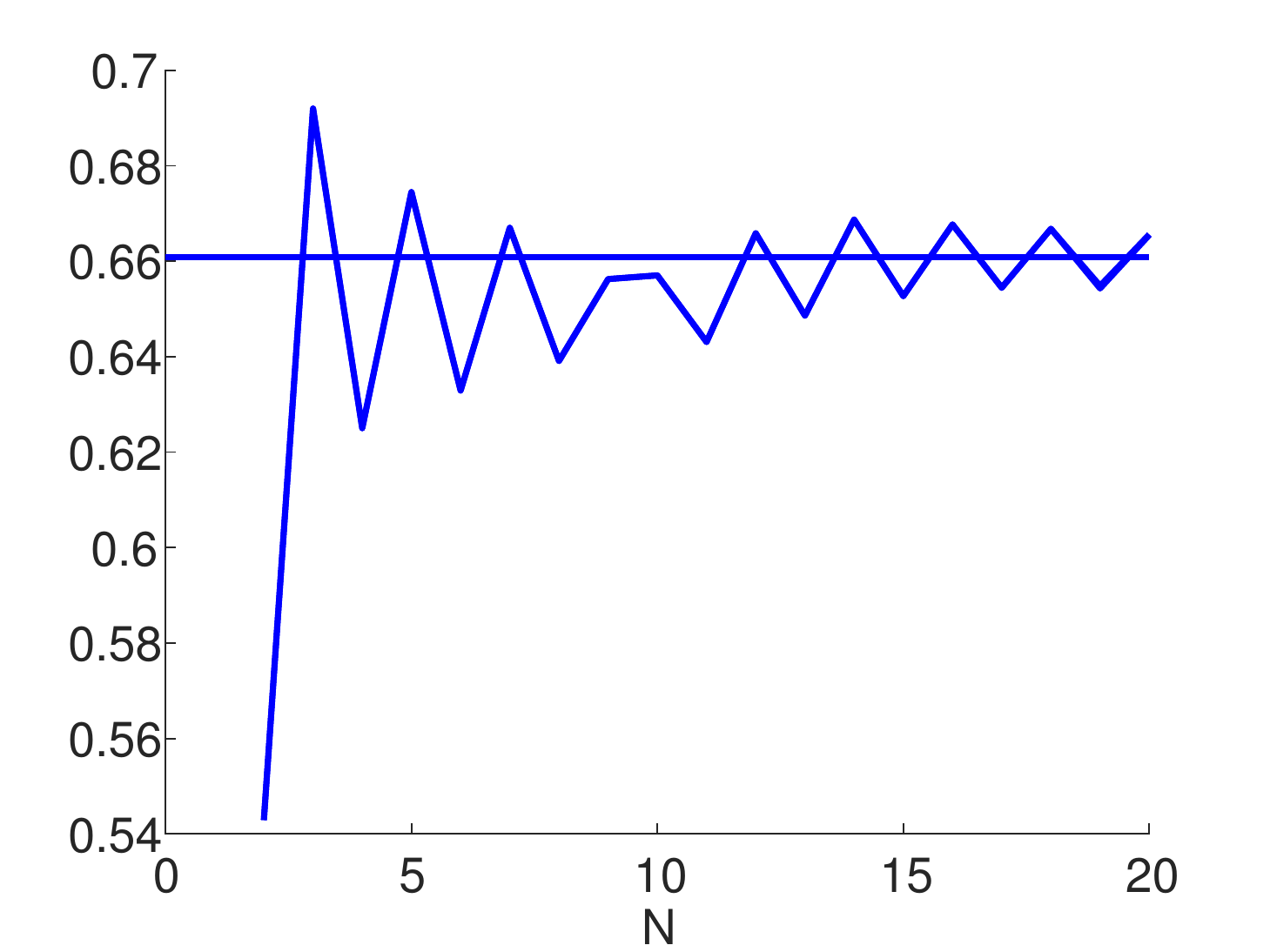}
\end{center}
\end{minipage}
\begin{minipage}{0.45\textwidth}
(b)
\begin{center}
 \includegraphics[scale=0.3]{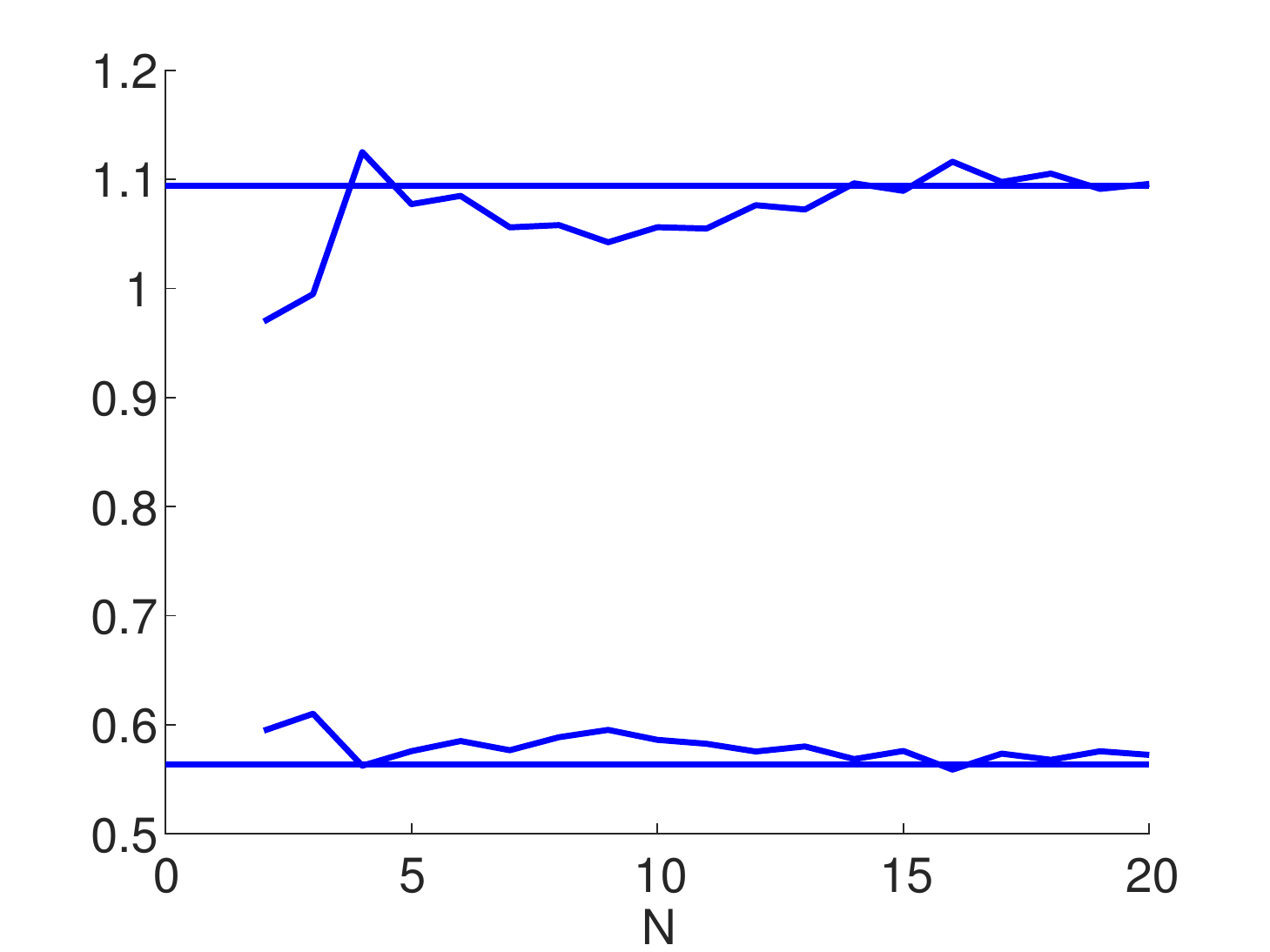}
\end{center}
\end{minipage}
	\caption{Upper and lower hedging prices of the European option \eqref{double_butterfly_f} with move set (a) $\chi_1$ and (b) $\chi_2$. Horizontal lines represent the limiting values. In (a), the upper and lower hedging prices coincide}
	\label{double_butterfly_price}
\end{figure}

\section{Conclusion}
We investigated the upper hedging price of multivariate contingent claims from the viewpoint of the game-theoretic probability.
The pricing problem is reduced to a backward induction of an optimization over simplexes.
For European options with submodular or supermodular payoff functions such as the options on the maximum or the minimum of several assets,
this optimization is solved in closed form by using the Lov\'asz extension.
As the number of game rounds goes to infinity, the upper hedging price of European options converges to the solution of the Black-Scholes-Barenblatt equation.
For European options with submodular or supermodular payoff functions, 
the Black-Scholes-Barenblatt equation is reduced to the linear Black-Scholes equation and it is solved in closed form.
Numerical experiments showed the validity of the theoretical results.

We mainly restricted our attention to European options.
Extension to path-dependent payoff functions, including American options, is an important future problem.
For such payoff functions, the definition of submodularity or supermodularity seems not trivial.

Also, we mainly assumed that the move set is a product set \eqref{prod_chi}.
In particular, our results on submodular and supermodular payoff functions are based on the lattice binomial model.
Extension to general move sets is another interesting future problem.

The Black-Scholes-Barenblatt equation is a special case of time-dependent diffusion equation \citep{Hundsdorfer}.
Although we used a simple finite difference method for numerical solution, it would be interesting to investigate more effective numerical methods. 

\section*{Acknowledgment}
We thank Naoki Marumo and Kengo Nakamura for helpful comments.
This work was supported by JSPS KAKENHI Grant Numbers 16K12399 and 17H06569.

\appendix

\section{Number of candidates $\widetilde{\chi}$ in three dimension}

Here, we calculate the number of candidates $\widetilde{\chi}$ when $d=3$.
Let ${\rm int} (A)$ be the interior of a set $A$ and $A^c$ be the complement of $A$.

\begin{Lemma}\label{lem:3dcube}
Let 
\begin{equation*}
	T_1 = {\rm conv} \{ (0,0,0),(0,1,1),(1,1,0),(1,0,1) \}, \ T_2 = {\rm conv} \{ (0,0,1),(0,1,0),(1,0,0),(1,1,1) \}
\end{equation*}
be the regular tetrahedra in the cube $[0,1]^3$.
Consider a point $x \in [0,1]^3$ and let $N(x) = \{ \{ z_0,z_1,z_2,z_3 \} \mid z_k \in \{ 0,1 \}^3, x \in {\rm conv} \{z_0,z_1,z_2,z_3\}, \dim \conv \{ z_0, z_1, z_2, z_3\}=3 \}$ be the set 
of tetrahedra containing $x$.

\begin{itemize}
\item If $x \in {\rm int} (T_1 \cap T_2)$, then $|N (x)| = 14$.
\item If $x \in {\rm int} (T_1 \cap T^c_2)$ or $x \in {\rm int} (T^c_1 \cap T_2)$, then $|N (x)| = 11$.

\item If $x \in {\rm int} (T^c_1 \cup T^c_2)$, then $|N (x)| = 8$.
\end{itemize}
\end{Lemma}

\begin{figure}
\begin{minipage}{0.4\textwidth}
\begin{center}
 \includegraphics[scale=0.3]{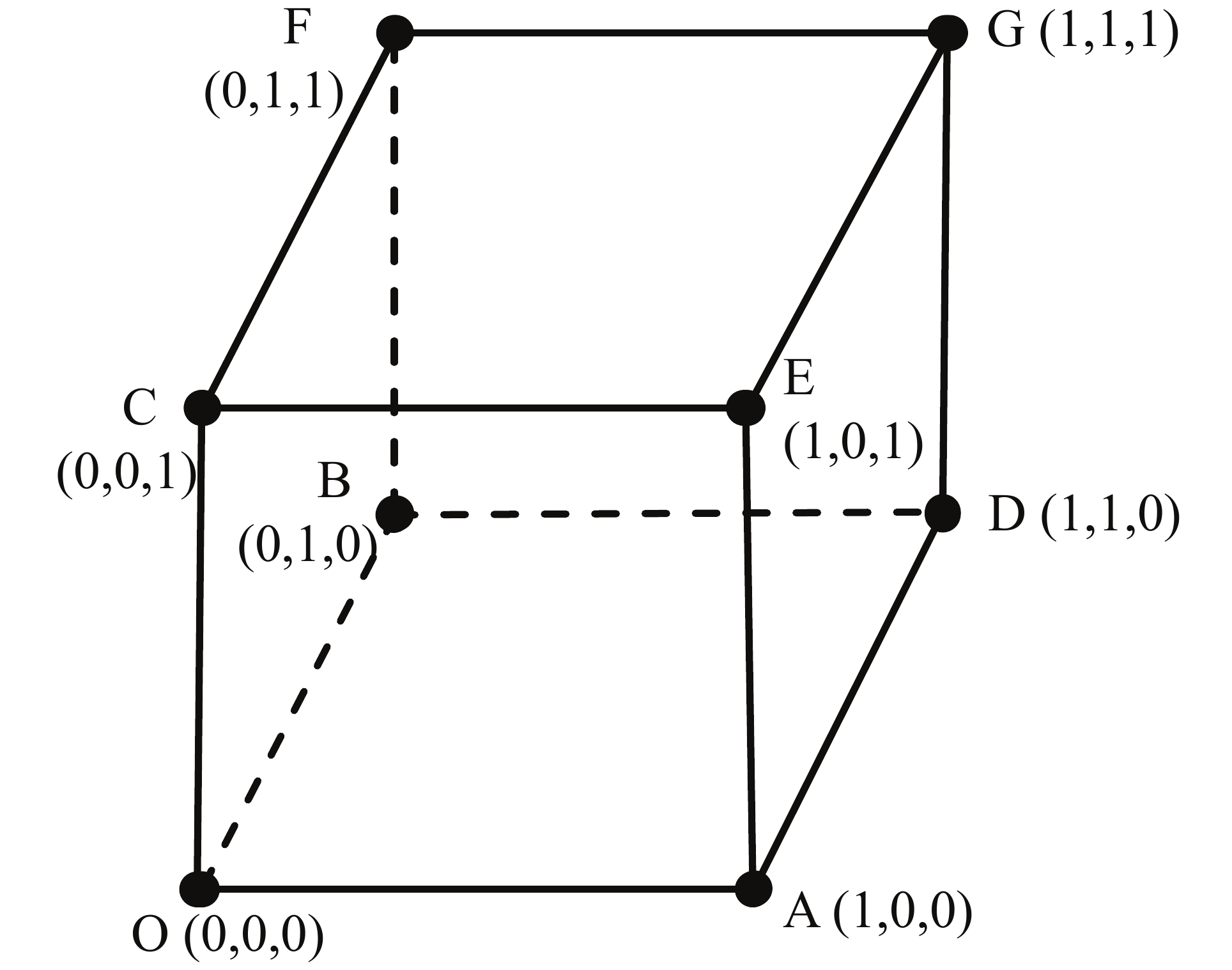}
\end{center}
\end{minipage}
\begin{minipage}{0.4\textwidth}
\begin{center}
 \includegraphics[scale=0.28]{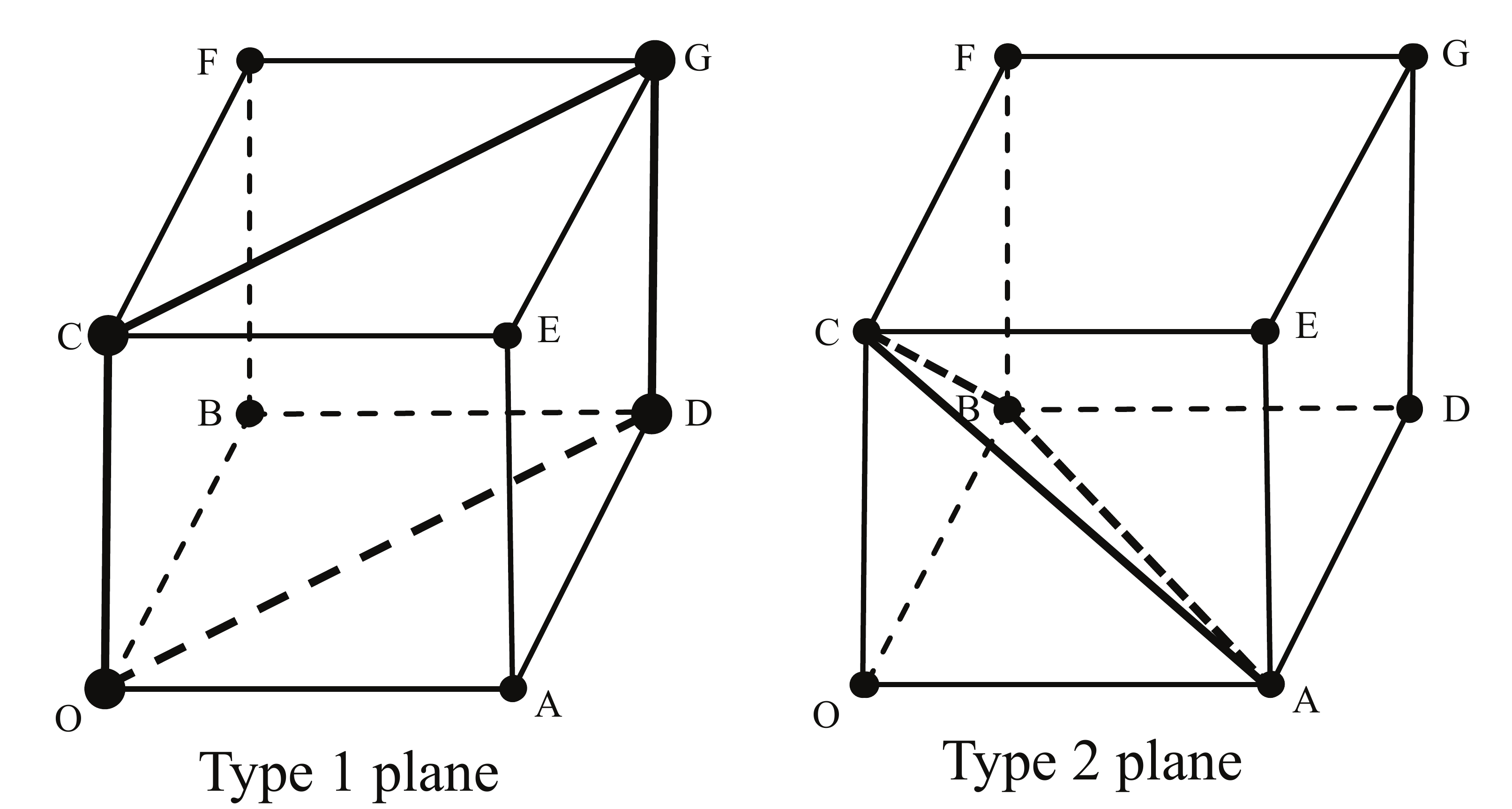}
\end{center}
\end{minipage}
\caption{cube and cutting planes}
\label{cube1}
\end{figure}

\begin{proof}
We only give a sketch of a proof, because we used computer to count the
number of tetrahedra containing a point $x$.  

The cube $[0,1]^3$ is denoted as the left picture of Figure \ref{cube1}.
There are 14 planes containing three or four vertices of the cube, which cut into the cube.
There are 6 Type 1 planes containing four vertices with the equations
\begin{equation}
\label{eq:6-planes}
  x=y, \ y=z,\ x=z, \ x+y=1, \ y+z=1, \ x+z=1,
\end{equation}
and there are 8 Type 2 planes containing three vertices with the equations
\begin{align}
  &x+y+z=1, \ x+y+z=2, \ x-y+z=0, \ x-y+z=1, \nonumber \\
  &x+y-z=0, \ x+y-z=1, \ -x+y+z=1, \ -x+y+z=0.
\label{eq:8planes}
\end{align}

There are 58 tetrahedra (simplexes) in 4 types as in Figure \ref{cube2}.
There are 8 Type 1 tetrahedra, which are corner simplexes.
There are 2 Type 2 regular tetrahedra denoted as $T_1, T_2$ in the lemma.
There are 24 Type 3 tetrahedra and there are 24 Type 4 tetrahedra.
We keep the list of these 58 tetrahedra in a computer program.

\begin{figure}
\begin{minipage}{0.9\textwidth}
\begin{center}
\includegraphics[scale=0.23]{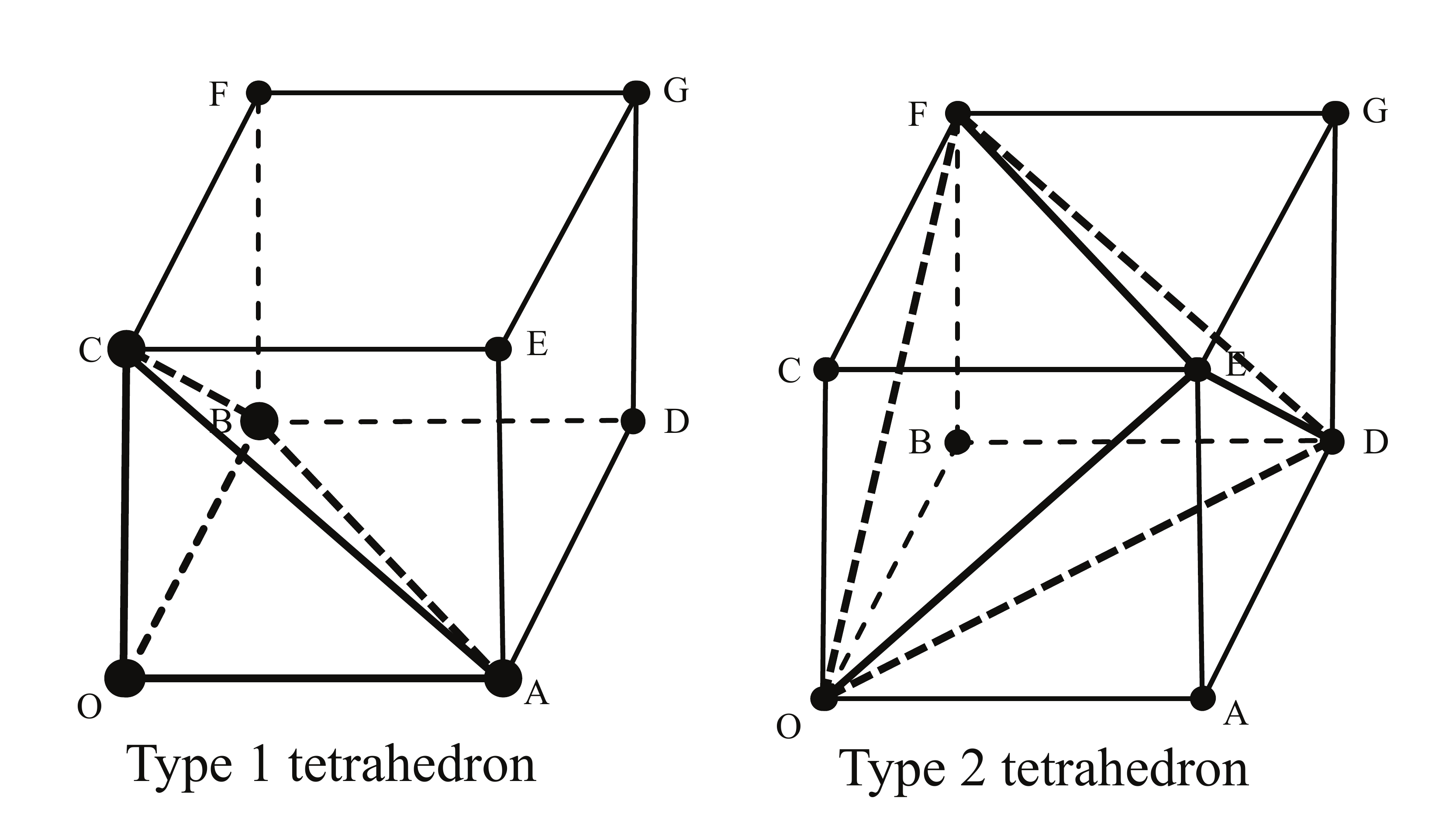}\includegraphics[scale=0.23]{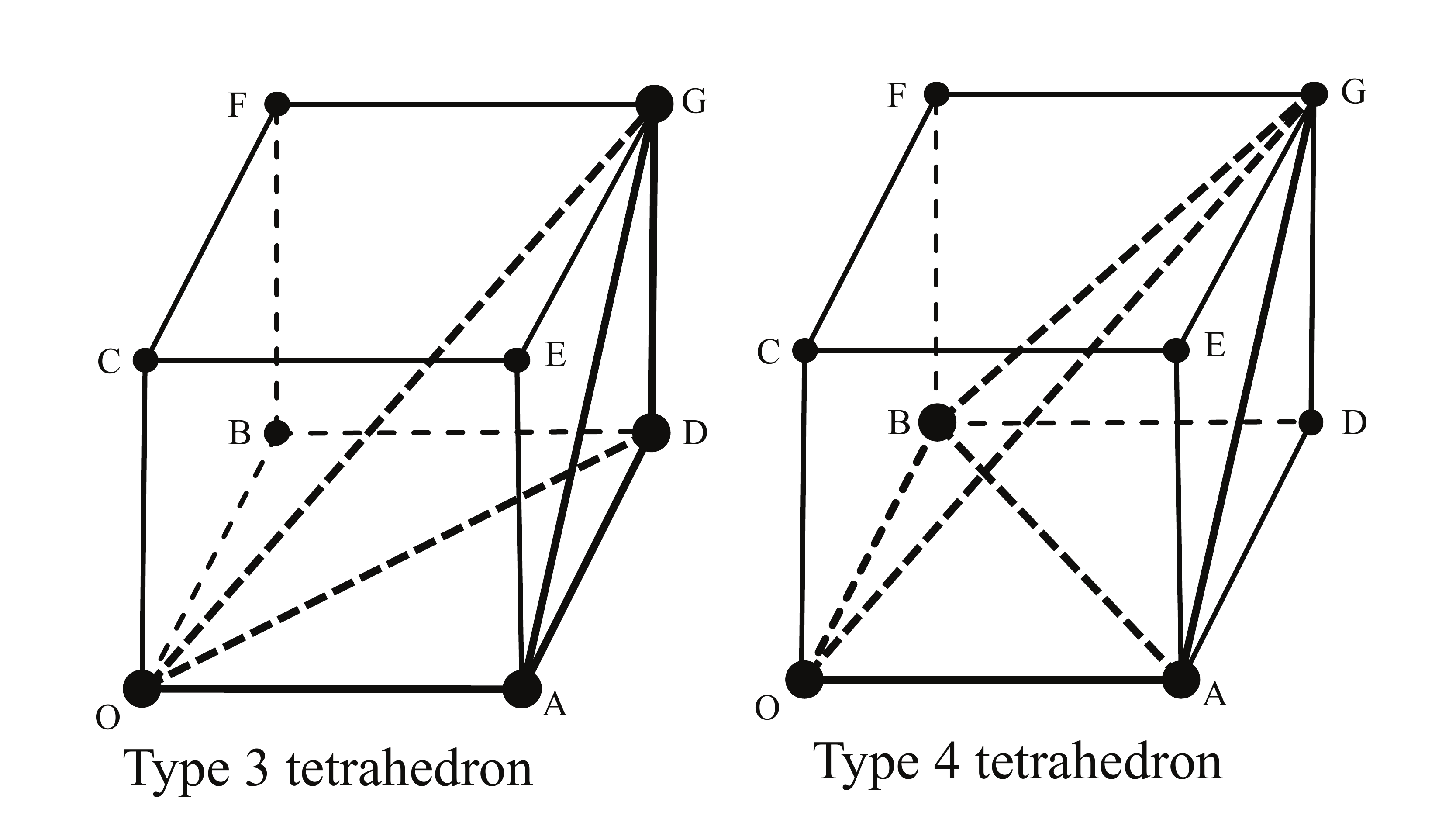}
\end{center}
\end{minipage}
\caption{Four types of tetrahedra}
\label{cube2}
\end{figure}

On the other hand, it is easy to visualize the 14 planes in 
\eqref{eq:6-planes} and \eqref{eq:8planes} by fixing $z$ and drawing
the sections of the cube as in Figure \ref{cube3}.  The 14 planes appear
as 14 lines inside the unit squares in Figure \ref{cube3}. Note that we only
need to consider $z<1/2$ by symmetry: $z\leftrightarrow 1-z$.
For each region of the sections of Figure \ref{cube3} we count the number
of tetrahedra containing the region.   Then we obtain the lemma.

\begin{figure}[htbp]
\begin{minipage}{0.9\textwidth}
\begin{center}
\includegraphics[scale=0.4]{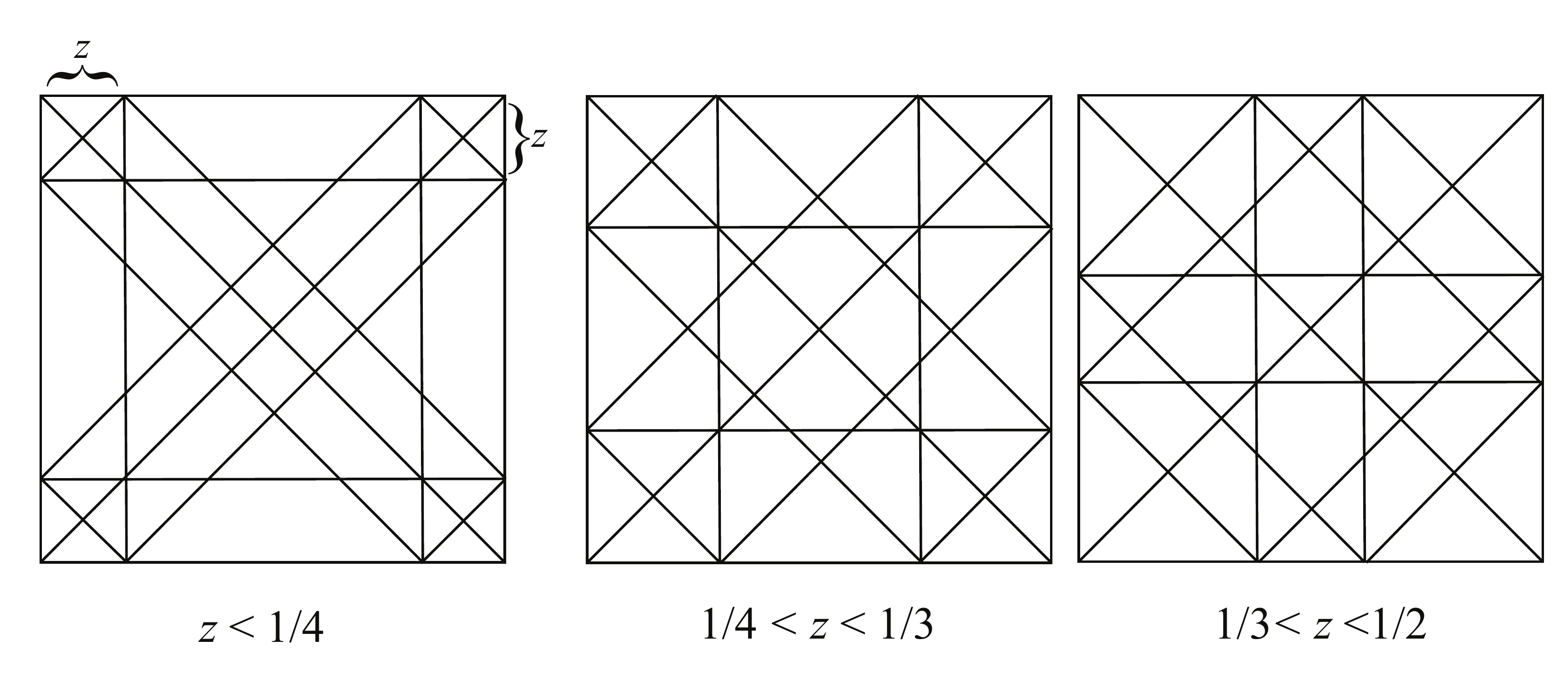}
\end{center}
\end{minipage}
\caption{Sections of the cube}
\label{cube3}
\end{figure}

\end{proof}

Note that Lemma \ref{lem:3dcube} applies only for generic $x$.  
For $x$ on the boundary of a tetrahedron, the number $|N(x)|$ may be larger.  
For example $x=(1/2,1/2,1/2)$ is contained in $|N(x)|=50$ tetrahedra of Types 2-4.

\section{Lower bound on the number of candidate $\widetilde{\chi}$}

Here, we provide a lower bound on the number of candidates $\widetilde{\chi}$ for general $d$.

\begin{Lemma}
Assume $d \geq 2$.
Consider a point $x$ in the $d$-dimensional hypercube $[0,1]^d$ and define $N (x) = \{ \{ z_0,\cdots,z_d \} \mid z_k \in \{ 0,1 \}^d, x \in {\rm conv} \{z_0,\cdots,z_d\}, \dim\conv \{z_1,\cdots,\allowbreak z_d\}=d\}$.
Then,
\begin{equation*}
	|N (x)| \geq 2^{d-2}.
\end{equation*}
\end{Lemma}
\begin{proof}
First, we consider the case where $x$ is a generic point
Without loss of generality, we assume 
\begin{equation}
	0 < x_1 < x_2 < \cdots < x_d < 1/2. \label{eq:x}
\end{equation}
Let
\[
	c = \min\{x_1, x_2 - x_1, x_3- x_2, \cdots, x_d - x_{d-1}\} > 0,
\]
and $i_{*}$ be an index attaining the minimum:
\begin{equation}
	c = x_{i_{*}} - x_{i_{*}-1}. \label{eq:i}
\end{equation}
Note that $c<1/4$ since $d \geq 2$.
In the following, we focus on simplexes that have the zero vector as one vertex.
Note that such simplex has nonempty interior if and only if the other $d$ vertices $y_1,\cdots,y_d$ are linearly independent.

Let $e_i \in \{ 0,1 \}^d$, $i=1,\cdots,d+1$ be a 0-1 vector defined by
\begin{equation}
	(e_i)_j=\begin{cases} 0 & (j < i) \\ 1 & (j \geq i) \end{cases}. \label{eq:3}
\end{equation}
In particular, $e_{d+1}$ is the zero vector.
Then, the vector $x$ is decomposed as
\begin{align}
	x = \sum_{i=1}^{d+1} (x_i-x_{i-1}) e_i, \label{eq:lovasz}
\end{align}
where we define $x_0=0$ and $x_{d+1}=1$.
Therefore, from \eqref{eq:x},
\begin{equation}
	x \in {\rm conv}\{ e_1,\cdots,e_{d+1} \}. \label{eq:conv}
\end{equation}

Now, we construct simplexes containing $x$ by changing the vertex $e_{i_{*}}$ in \eqref{eq:conv}. 
Note that every 0-1 vector $z \in \{ 0,1 \}^d$ is uniquely expressed as
\[
	z = \sum_{i=1}^d c_i e_i,
\]
where $c_i \in \{ -1,0,1 \}$ satisfies 
\[
	\sum_{j=1}^i c_j \in \{ 0,1 \} \quad (i=1,\cdots,d).
\]
Under this correspondence, there are $2^{d-2}$ vectors $\epsilon$ with $c_{i_{*}-1} = 0$ and $c_{i_{*}} = 1$.
These vectors are given by
\begin{equation}
	\epsilon = \cdots - e_k + e_{i_{*}} - e_l + \cdots \label{eq:4}
\end{equation}
for some $k < i_{*}$ and $l > i_{*}$.
Thus,
\begin{equation}
	\epsilon  - \cdots + e_k +e_l - \cdots = e_{i_{*}}, \label{eq:6}
\end{equation}
where the sum of coefficients in the left hand side is 1 or 2.
By substituting \eqref{eq:6} into \eqref{eq:lovasz},
\begin{align*}
	x &= \sum_{i \neq i_{*}} (x_i-x_{i-1}) e_i + c (\epsilon  - \cdots + e_k +e_l - \cdots) \\
	&= \sum_{i \neq i_{*}} b_i e_i + c \epsilon,
\end{align*}
where $b_i \geq 0$ and $\sum_{i \neq i_{*}} b_i + c=1$.
Therefore,
\[
	x \in {\rm conv} \{ e_1,\cdots,e_{i_{*}-1},\epsilon,e_{i_{*}+1},\cdots,e_{d+1} \}. 
\]
Since there are $2^{d-2}$ choices of $\epsilon$, we have $|N (x)| \geq 2^{d-2}$.

Next, we consider the case where $x$ is not a generic point.
Without loss of generality, we assume
\begin{equation*}
	0 \leq x_1 \leq x_2 \leq \cdots \leq x_d \leq 1/2,
\end{equation*}
where at least one inequality holds with equality.
Then, we can take a sequence of generic points $y_1,y_2,\cdots$ satisfying \eqref{eq:x} and having the same $i_*$ that converges to $x$.
Let $N=\bigcap N (y_k)$.
Then, from the above argument, $|N (y_k)| \geq 2^{d-2}$ for each $k$ and the
$2^{d-2}$ simplexes containing $y_k$ are common.  In particular $|N| \geq 2^{d-2}$.
Also, since the simplex is closed, $x$ belongs to each simplex in $N$: $N \subset N(x)$.
Therefore, $|N (x)| \geq 2^{d-2}$.
\end{proof}

\end{document}